\documentclass{LMCS}

\def\dOi{9(4:24)2013}
\lmcsheading%
{\dOi}
{1--25}
{}
{}
{Jan.~29, 2013}
{Dec.~18, 2013}
{}


\ACMCCS{[{\bf Theory of computation}]: Logic---Proof theory; Formal
  languages and automata theory---Tree languages}

\usepackage{proof}
\usepackage{stmaryrd}
\usepackage{amsmath}
\usepackage{amssymb}
\usepackage{amsthm}
\usepackage{latexsym}
\usepackage{mathdots}
\usepackage{mathabx}
\usepackage{helvet}
\usepackage{calrsfs}
\usepackage{color}
\usepackage{fixltx2e}
\usepackage{url}

\usepackage[pdftex]{virginialake}
\vlnosmallleftlabels
\vlnostructuressyntax
\newdimen\myvltreesize
\renewcommand{\vltrauxx}[5]{\vldaux
   {#1}
   {#2}
   {\vlhy{#3}}
   {}
   {\hbox{$\vcenter{\xy
          0;<\hsize,0pt>:<0pt,#4\hsize>::
          (-0.5,0.5);(0.5,0.5)**\crv{(0.1,0.6)&(-0.1,0.4)};
          (0,0)**@{-};(-0.5,0.5)**@{-};
          (0,0.293)*{#5\strut}
          \endxy}$}}
   {\kern\deropen}}

\newcommand{\vlstrf}[4]{\vltrf{\hbox{\small$#1$}}{#2}{\vlshy{}}{#3}{\vlshy{}}{#4}}
\newcommand{\vlbtrf}[4]{\vlstrf{#1}{#2}{\vlshy{\hskip#3}}{#4}}
\newcommand{\vlhtr}[2]{\vlbtrf{#1}{#2}{2em}{1}}
\newcommand{\vlshy}[1]{\vlhyaux{$#1$}}

\newcommand{\vlfa}[1]{\forall #1\,}
\newcommand{\vlex}[1]{\exists #1\,}

\theoremstyle{plain}
\theoremstyle{definition}

\def\cA{{\mathcal A}}
\def\cB{{\mathcal B}}

\def\cI{{\mathcal I}}

\def\set#1{\{#1\}}

\def\tuple#1{\langle#1\rangle}

\def\grammareq {\mathrel{\raise.4pt\hbox{::}{=}}}%

\newcommand{\dotto}[1][]{\mathrel{\!\xy\ar@{.>}^-{#1}(5,0)\endxy\!}}
\newcommand{\solto}[1][]{\mathrel{\!\xy\ar@{->}^-{#1}(5,0)\endxy\!}}
\newcommand{\longsolto}[1][]{\mathrel{\!\xy\ar@{->}^-{#1}(11,0)\endxy\!}}
\newcommand{\longdotto}[1][]{\mathrel{\!\xy\ar@{.>}^-{#1}(11,0)\endxy\!}}
\newcommand{\xldotto}[2][]{\mathrel{\!\xy\ar@{.>}^-{#1}(#2,0)\endxy\!}}

\def\cutr{\mathsf{cut}}

\def\axr{\mathsf{ax}}

\def\weakr{\mathsf{weak}}
\def\conr{\mathsf{cont}}

\def\rr{\mathsf{r}}

\def\wcneg#1{\overline{#1}}

\newbox\cutbox
\newdimen\cutwd
\newdimen\cutht
\newdimen\cutdp
\def\ccut{%
  \setbox\cutbox\hbox{$\lozenge$}
  \cutwd=\wd\cutbox
  \cutht=\ht\cutbox
  \cutdp=\dp\cutbox
  \setbox\cutbox\hbox to\cutwd{\hss\vrule width.3pt height\cutht depth\cutdp\hss}
  \mathbin{\lozenge\hskip-\cutwd\copy\cutbox}}
\def\scriptcut{%
  \setbox\cutbox\hbox{$\scriptstyle\lozenge$}
  \cutwd=\wd\cutbox
  \cutht=\ht\cutbox
  \cutdp=\dp\cutbox
  \setbox\cutbox\hbox to\cutwd{\hss\vrule width.3pt height\cutht depth\cutdp\hss}
  \mathord{\lozenge\hskip-\cutwd\copy\cutbox}}

\def\vccut{%
  \setbox\cutbox\hbox{$\lozenge$}
  \cutwd=\wd\cutbox
  \cutht=\ht\cutbox
  \cutdp=\dp\cutbox
  \setbox\cutbox\hbox to\cutwd{\hss\hskip.3pt\vrule width.3pt height\cutht depth\cutdp\hss}
  \mathbin{\lozenge\hskip-\cutwd\copy\cutbox}}

\def\lrgldel {\mathchoice{(}{(}{\langle}{\langle}}%
\def\lrgrdel {\mathchoice{)}{)}{\rangle}{\rangle}}%
\def\aprldel {\mathchoice
    {\mathopen {\setbox0=\hbox{$\displaystyle     \lrgldel$}\hbox to\wd0
                         {\hfil$\displaystyle     (       $\hfil}}}%
    {\mathopen {\setbox0=\hbox{$\textstyle        \lrgldel$}\hbox to\wd0
                         {\hfil$\textstyle        (        $\hfil}}}%
    {\mathopen {\setbox0=\hbox{$\scriptstyle      \lrgldel$}\hbox to\wd0
                         {\hfil$\scriptstyle      (        $\hfil}}}%
    {\mathopen {\setbox0=\hbox{$\scriptscriptstyle\lrgldel$}\hbox to\wd0
                         {\hfil$\scriptscriptstyle(        $\hfil}}}}%
\def\aprrdel {\mathchoice
    {\mathclose{\setbox0=\hbox{$\displaystyle     \lrgrdel$}\hbox to\wd0
                         {\hfil$\displaystyle     )       $\hfil}}}%
    {\mathclose{\setbox0=\hbox{$\textstyle        \lrgrdel$}\hbox to\wd0
                         {\hfil$\textstyle        )        $\hfil}}}%
    {\mathclose{\setbox0=\hbox{$\scriptstyle      \lrgrdel$}\hbox to\wd0
                         {\hfil$\scriptstyle      )        $\hfil}}}%
    {\mathclose{\setbox0=\hbox{$\scriptscriptstyle\lrgrdel$}\hbox to\wd0
                         {\hfil$\scriptscriptstyle)        $\hfil}}}}%
\def\seqldel {\mathchoice
    {\mathopen {\setbox0=\hbox{$\displaystyle     \lrgldel$}\hbox to\wd0
                         {\hfil$\displaystyle     \langle  $\hfil}}}%
    {\mathopen {\setbox0=\hbox{$\textstyle        \lrgldel$}\hbox to\wd0
                         {\hfil$\textstyle        \langle  $\hfil}}}%
    {\mathopen {\setbox0=\hbox{$\scriptstyle      \lrgldel$}\hbox to\wd0
                         {\hfil$\scriptstyle      \langle  $\hfil}}}%
    {\mathopen {\setbox0=\hbox{$\scriptscriptstyle\lrgldel$}\hbox to\wd0
                         {\hfil$\scriptscriptstyle\langle  $\hfil}}}}%
\def\seqrdel {\mathchoice
    {\mathclose{\setbox0=\hbox{$\displaystyle     \lrgrdel$}\hbox to\wd0
                         {\hfil$\displaystyle     \rangle  $\hfil}}}%
    {\mathclose{\setbox0=\hbox{$\textstyle        \lrgrdel$}\hbox to\wd0
                         {\hfil$\textstyle        \rangle  $\hfil}}}%
    {\mathclose{\setbox0=\hbox{$\scriptstyle      \lrgrdel$}\hbox to\wd0
                         {\hfil$\scriptstyle      \rangle  $\hfil}}}%
    {\mathclose{\setbox0=\hbox{$\scriptscriptstyle\lrgrdel$}\hbox to\wd0
                         {\hfil$\scriptscriptstyle\rangle  $\hfil}}}}%
\def\parldel {\mathchoice
    {\mathopen {\setbox0=\hbox{$\displaystyle     \lrgldel$}\hbox to\wd0
                         {\hfil$\displaystyle     [       $\hfil}}}%
    {\mathopen {\setbox0=\hbox{$\textstyle        \lrgldel$}\hbox to\wd0
                         {\hfil$\textstyle        [        $\hfil}}}%
    {\mathopen {\setbox0=\hbox{$\scriptstyle      \lrgldel$}\hbox to\wd0
                         {\hfil$\scriptstyle      [        $\hfil}}}%
    {\mathopen {\setbox0=\hbox{$\scriptscriptstyle\lrgldel$}\hbox to\wd0
                         {\hfil$\scriptscriptstyle[        $\hfil}}}}%
\def\parrdel {\mathchoice
    {\mathclose{\setbox0=\hbox{$\displaystyle     \lrgrdel$}\hbox to\wd0
                         {\hfil$\displaystyle     ]       $\hfil}}}%
    {\mathclose{\setbox0=\hbox{$\textstyle        \lrgrdel$}\hbox to\wd0
                         {\hfil$\textstyle        ]        $\hfil}}}%
    {\mathclose{\setbox0=\hbox{$\scriptstyle      \lrgrdel$}\hbox to\wd0
                         {\hfil$\scriptstyle      ]        $\hfil}}}%
    {\mathclose{\setbox0=\hbox{$\scriptscriptstyle\lrgrdel$}\hbox to\wd0
                         {\hfil$\scriptscriptstyle]        $\hfil}}}}%

\def\eightpoint{\small}                         
\def\pluldel {\mathchoice
   {\mathopen {\setbox0=\hbox{$\displaystyle     \lrgldel$}\hbox to\wd0
                        {\hfil$\displaystyle     [       $\hfil}%
                        \kern-\wd0\hbox to\wd0
                        {\hss$\vcenter{\hbox{\eightpoint$\scriptscriptstyle\bullet$}}$\hss}}}%
   {\mathopen {\setbox0=\hbox{$\textstyle        \lrgldel$}\hbox to\wd0
                        {\hfil$\textstyle        [       $\hfil}%
                        \kern-\wd0\hbox to\wd0
                        {\hss$\vcenter{\hbox{\eightpoint$\scriptscriptstyle\bullet$}}$\hss}}}%
   {\mathopen {\setbox0=\hbox{$\scriptstyle      \lrgldel$}\hbox to\wd0
                        {\hfil$\scriptstyle      [       $\hfil}%
                        \kern-\wd0\hbox to\wd0
                        {\hss$\raise.1ex\hbox{\eightpoint$\scriptscriptstyle\bullet$}$\hss}}}%
   {\mathopen {\setbox0=\hbox{$\scriptscriptstyle\lrgldel$}\hbox to\wd0
                        {\hfil$\scriptscriptstyle[       $\hfil}%
                        \kern-\wd0\hbox to\wd0
                        {\hss$\raise.03ex\hbox{\eightpoint$\scriptscriptstyle\bullet$}$\hss}}}}%
\def\plurdel {\mathchoice
   {\mathclose{\setbox0=\hbox{$\displaystyle     \lrgldel$}\hbox to\wd0
                        {\hfil$\displaystyle     ]       $\hfil}%
                        \kern-\wd0\hbox to\wd0
                        {\hss$\vcenter{\hbox{\eightpoint$\scriptscriptstyle\bullet$}}$\hss}}}%
   {\mathclose{\setbox0=\hbox{$\textstyle        \lrgldel$}\hbox to\wd0
                        {\hfil$\textstyle        ]       $\hfil}%
                        \kern-\wd0\hbox to\wd0
                        {\hss$\vcenter{\hbox{\eightpoint$\scriptscriptstyle\bullet$}}$\hss}}}%
   {\mathclose{\setbox0=\hbox{$\scriptstyle      \lrgldel$}\hbox to\wd0
                        {\hfil$\scriptstyle      ]       $\hfil}%
                        \kern-\wd0\hbox to\wd0
                        {\hss$\raise.1ex\hbox{\eightpoint$\scriptscriptstyle\bullet$}$\hss}}}%
   {\mathclose{\setbox0=\hbox{$\scriptscriptstyle\lrgldel$}\hbox to\wd0
                        {\hfil$\scriptscriptstyle]       $\hfil}%
                        \kern-\wd0\hbox to\wd0
                        {\hss$\raise.03ex\hbox{\eightpoint$\scriptscriptstyle\bullet$}$\hss}}}}%
\def\witldel {\mathchoice
   {\mathopen {\setbox0=\hbox{$\displaystyle     \lrgldel$}\hbox to\wd0
                        {\hfil$\displaystyle     (       $\hfil}%
                        \kern-\wd0\hbox to\wd0
                        {\hss$\vcenter{\hbox{\eightpoint$\scriptscriptstyle\bullet\mkern3.2mu$}}$\hss}}}%
   {\mathopen {\setbox0=\hbox{$\textstyle        \lrgldel$}\hbox to\wd0
                        {\hfil$\textstyle        (       $\hfil}%
                        \kern-\wd0\hbox to\wd0
                        {\hss$\vcenter{\hbox{\eightpoint$\scriptscriptstyle\bullet\mkern3.2mu$}}$\hss}}}%
   {\mathopen {\setbox0=\hbox{$\scriptstyle      \lrgldel$}\hbox to\wd0
                        {\hfil$\scriptstyle      (       $\hfil}%
                        \kern-\wd0\hbox to\wd0
                        {\hss$\raise.1ex\hbox{\eightpoint$\scriptscriptstyle\bullet\mkern3.2mu$}$\hss}}}%
   {\mathopen {\setbox0=\hbox{$\scriptscriptstyle\lrgldel$}\hbox to\wd0
                        {\hfil$\scriptscriptstyle(       $\hfil}%
                        \kern-\wd0\hbox to\wd0
                        {\hss$\raise.03ex\hbox{\eightpoint$\scriptscriptstyle\bullet\mkern3.2mu$}$\hss}}}}%
\def\witrdel {\mathchoice
   {\mathclose{\setbox0=\hbox{$\displaystyle     \lrgldel$}\hbox to\wd0
                        {\hfil$\displaystyle     )       $\hfil}%
                        \kern-\wd0\hbox to\wd0
                        {\hss$\vcenter{\hbox{\eightpoint$\scriptscriptstyle\mkern3.2mu\bullet$}}$\hss}}}%
   {\mathclose{\setbox0=\hbox{$\textstyle        \lrgldel$}\hbox to\wd0
                        {\hfil$\textstyle        )       $\hfil}%
                        \kern-\wd0\hbox to\wd0
                        {\hss$\vcenter{\hbox{\eightpoint$\scriptscriptstyle\mkern3.2mu\bullet$}}$\hss}}}%
   {\mathclose{\setbox0=\hbox{$\scriptstyle      \lrgldel$}\hbox to\wd0
                        {\hfil$\scriptstyle      )       $\hfil}%
                        \kern-\wd0\hbox to\wd0
                        {\hss$\raise.1ex\hbox{\eightpoint$\scriptscriptstyle\mkern3.2mu\bullet$}$\hss}}}%
   {\mathclose{\setbox0=\hbox{$\scriptscriptstyle\lrgldel$}\hbox to\wd0
                        {\hfil$\scriptscriptstyle)       $\hfil}%
                        \kern-\wd0\hbox to\wd0
                        {\hss$\raise.03ex\hbox{\eightpoint$\scriptscriptstyle\mkern3.2mu\bullet$}$\hss}}}}%

\newbox\ldelbox
\setbox\ldelbox=\hbox{$\lrgldel$}

\newbox\rdelbox
\setbox\rdelbox=\hbox{$\lrgrdel$}

\def\quadfs {\rlap{\rm\quad.}}%
\def\qqualto {\qquad\leadsto\qquad}%
\def\qqquand {\quad\qquad\mbox{and}\qquad\quad}%
\def\nqquand {\qquad\mbox{\normalsize and}\qquad}%
\def\qqquor {\quad\qquad\mbox{or}\qquad\quad}%
\def\nqquor {\qquad\mbox{\normalsize or}\qquad}%
\def\qqquiff {\quad\qquad\mbox{iff}\qquad\quad}%

\def\qcdots{\quad\cdots\quad}

\def\clap#1{\hbox to 0pt{\hss#1\hss}}
\def\sqlap#1{\hbox to .5em{\hss#1\hss}}
\def\qlap#1{\hbox to 1em{\hss#1\hss}}
\def\qqlap#1{\hbox to 2em{\hss#1\hss}}
\def\qqqlap#1{\hbox to 3em{\hss#1\hss}}
\def\qqqqlap#1{\hbox to 4em{\hss#1\hss}}
\def\qqqqqlap#1{\hbox to 5em{\hss#1\hss}}
\def\qqqqqqlap#1{\hbox to 6em{\hss#1\hss}}
\def\qqqqqqqlap#1{\hbox to 7em{\hss#1\hss}}
\def\qqqqqqqqlap#1{\hbox to 8em{\hss#1\hss}}
\def\qqqqqqqqqlap#1{\hbox to 9em{\hss#1\hss}}
\newcommand{\wlap}[2][10ex]{\hbox to#1{\hss#2\hss}}

\newcommand{\wlapm}[2][10ex]{\hbox to#1{\hss$#2$\hss}}
\def\rlapm#1{\hbox to 0pt{$#1$\hss}}
\def\llapm#1{\hbox to 0pt{\hss$#1$}}

\def\qqquad{\quad\qquad}

\newcommand{\vclap}[2][0pt]{\hbox to #1{\hss#2\hss}}
\newcommand{\vclapm}[2][0pt]{\hbox to #1{\hss$#2$\hss}}

\def\semiproofadjust{\vadjust{\nobreak\vskip-1.3ex\nobreak}}

\def\interdisplayskip{.5ex}
\newskip\mydisplaywidth
\newcommand{\twolinedisplay}[3][10pt]{%
  \mydisplaywidth=\displaywidth
  \advance\mydisplaywidth-#1
  \begin{array}{c}
    \clap{\hbox to\mydisplaywidth{$\displaystyle#2$\hss}}\\[\interdisplayskip]
    \clap{\hbox to\mydisplaywidth{\hss$\displaystyle#3$}}
  \end{array}
}

\usepackage{accents}
\newlength{\dhatheight}
\newcommand{\doublehat}[1]{%
    \settoheight{\dhatheight}{\ensuremath{\hat{#1}}}%
    \addtolength{\dhatheight}{-0.35ex}%
    \hat{\vphantom{\rule{1pt}{\dhatheight}}%
    \smash{\hat{#1}}}}

\newcommand{\union}{\cup}
\newcommand{\Union}{\bigcup}

\newcommand{\sop}{[}
\newcommand{\scl}{]}
\newcommand{\sel}[2]{#1 \backslash #2}
\newcommand{\unsubst}[2]{\sop \sel{#1}{#2} \scl}
\newcommand{\Pos}{\mathrm{Pos}} 

\newcommand{\Lang}[1]{L(#1)}	
\newcommand{\Gram}[1]{\mathrm{G}(#1)}	
\newcommand{\Gramreg}[1]{\mathrm{G}_0(#1)}	

\newcommand{\TR}[1]{\mathcal{T}_{#1}}
\newcommand{\TRV}[2]{\mathcal{T}_{#1}(#2)}
\newcommand{\DD}{\mathcal{D}}
\newcommand{\DDt}{\tilde{\mathcal{D}}}
\newcommand{\PP}{\mathcal{P}}

\newcommand{\derel}[1]{\mathbin{\rightarrow_{#1}}}
\newcommand{\sderel}[1]{\mathbin{\rightarrow_{#1}^*}}
\newcommand{\atpos}[2]{#1|_{#2}}
\newcommand{\sreach}[1]{\mathbin{\prec_{#1}^*}}
\newcommand{\preach}[1]{\mathbin{\prec_{#1}^+}}
\newcommand{\nreach}[1]{\mathbin{\prec_{#1}^{}}}
\newcommand{\reach}[1]{\mathbin{\prec_{#1}}}

\newcommand{\nonrig}[2]{\mathbf{n}(#1,#2)}
\newcommand{\nonrigd}[1]{\mathbf{n}(#1)}

\newcommand{\pathto}[1][]{\mathbin{\stackrel{#1}{\dashrightarrow}}}
\newcommand{\length}[1]{\left|#1\right|}
\newcommand{\Hcont}[1]{[\mkern-3mu[#1]\mkern-3mu]}

\newcommand{\scutred}{\leadsto}

\newcommand{\dual}[1]{\wcneg{#1}}
\newcommand{\Lor}{\bigvee}

\newcommand{\EV}{\mathrm{EV}}	
\newcommand{\EVc}{\mathrm{EV_c}}	
\newcommand{\Hseq}{\mathrm{H}}	
\newcommand{\Bsub}{\mathrm{B}}	
\newcommand{\terms}[1]{\mathrm{tm}(#1)}	
\newcommand{\cred}{\rightsquigarrow}	
\newcommand{\credm}{\stackrel{\mathit{ne}}{\rightsquigarrow}}	
\newcommand{\genrel}{\longrightarrow}
\newcommand{\sk}{\mathrm{sk}}

\newcommand{\lkcut}[3][]{\infer[\cutr{#1}]{#2}{#3}}

\newcommand{\lku}[3][]{\infer[\forall{#1}]{#2}{#3}}

\def\weakr{\mathsf{w}}
\def\conr{\mathsf{c}}

\newcommand{\forallrr}[1]{\forall_{#1}}
\newcommand{\existsrr}[1]{\exists_{#1}}
\newcommand{\cutrr}[1]{\cutr_{#1}}

\newcommand{\leftabove}[2]{#1\mathbin{\Lsh}#2}
\newcommand{\rightabove}[2]{#2\mathbin{\Rsh}#1}
\newcommand{\rparallel}[2]{#1\mathbin{\Lsh\Rsh}#2}

\newcommand{\Hsub}{\le}
\newcommand{\Hsup}{\ge}

\newcommand{\deprelz}{\prec}

\newcommand{\deprel}{\llcurly}

\title[Herbrand-Confluence]{Herbrand-Confluence\rsuper*}

\author[S.~Hetzl]{Stefan Hetzl\rsuper a}
\address{{\lsuper a}Institute of Discrete Mathematics and Geometry\\
  Vienna University of Technology\\
  Wiedner Hauptstra{\ss}e 8-10\\
  1040 Vienna, Austria}
\email{stefan.hetzl@tuwien.ac.at}

\author[L.~Stra\ss burger]{Lutz Stra\ss burger\rsuper b}
\address{{\lsuper b}INRIA Saclay -- \^Ile-de-France\\
1 rue Honor\'e d'Estienne d'Orves\\
B\^atiment Alan Turing, Campus de l'\'Ecole Polytechnique\\
91120 Palaiseau, France}
\email{lutz@lix.polytechnique.fr}

\titlecomment{{\lsuper*}This paper is an extended version of~\cite{Hetzl12Herbrand}.}

\begin{document}

\begin{abstract}
We consider cut-elimination in the sequent calculus for classical
first-order logic. It is well known that this system, in its most
general form, is neither confluent nor strongly normalizing. In this
work we take a coarser (and mathematically more realistic) look at
cut-free proofs. We analyze which witnesses they choose for which
quantifiers, or in other words: we only consider the
Herbrand-disjunction of a cut-free proof. Our main theorem is a
confluence result for a natural class of proofs: all (possibly
infinitely many) normal forms of the non-erasing reduction lead to the
same Herbrand-disjunction.
\end{abstract}

\maketitle


\section{Introduction}

The constructive content of proofs has always been a central topic of
proof theory and it is also one of the most important influences that
logic has on computer science.  Classical logic is widely used and
presents interesting challenges when it comes to understanding the
constructive content of its proofs. These challenges have therefore
attracted considerable attention, see, for
example,~\cite{Parigot92LambdaMu,Danos97New,Curien00Duality},
\cite{Barbanera96Symmetric},
\cite{Urban00Classical,Urban01Strong},
\cite{Berger02Refined},
\cite{Kohlenbach08Applied}, or~\cite{Baaz00CutElimination} for
different investigations in this direction.

A well-known, but not yet well-understood, phenomenon is that a single
classical proof usually allows several different constructive
readings. From the point of view of applications this means that we
have a choice among different programs that can be extracted.
In~\cite{Ratiu12Exploring} the authors show that two different
extraction methods applied to the same proof produce two programs, one
of polynomial and one of exponential average-case complexity. This
phenomenon is further exemplified by case studies
in~\cite{Urban00Classical,Baaz05Experiments,Baaz08Ceres} as well as
the asymptotic
results~\cite{Baaz11Nonconfluence,Hetzl12Computational}. The reason
for this behavior is that classical ``proofs often leave algorithmic
detail underspecified''~\cite{Avigad10Computational}.

On the level of cut-elimination in the sequent calculus this
phenomenon is reflected by the fact that the standard proof reduction
without imposing any strategy is not confluent.  In this paper we
consider cut-elimination in classical first-order logic and treat the
question which cut-free proofs one can obtain (by the strategy-free
rewriting system) from a single proof with cuts. As our aim is to
compare cut-free proofs we need a notion of equivalence of proofs:
clearly the syntactic equality makes more differences than those which
are mathematically interesting. Being in a system with quantifiers, a
natural and more realistic choice is to consider two cut-free proofs
equivalent if they choose the same terms for the same quantifiers, in
other words: if they have the same Herbrand-disjunction.

A cut-reduction relation will then be called {\em Herbrand-confluent}
if all its normal forms have the same Herbrand-disjunction. The main
result of this paper is that, for a natural class of proofs, the
standard reduction without erasing of subproofs is
Herbrand-confluent. This result is surprising as this reduction is
neither confluent nor strongly normalizing and may produce normal
forms of arbitrary size (which---as our result shows---arise only
from repetitions of the same instances).

As a central proof technique we use rigid tree languages which have
been introduced in~\cite{Jacquemard09Rigid} with applications in
verification (e.g.\ of cryptographic protocols as
in~\cite{Jacquemard11Rigid}) as their primary purpose. To a proof we
will associate a rigid tree grammar whose language is invariant under
non-erasing cut-elimination and hence equal to the only obtainable
Herbrand-disjunction. This property suggests the new notion of
\emph{Herbrand-content} of a proof, which is defined as the language of the
grammar of the proof, and which is a strong invariant.  A side effect of this
proof technique is a combinatorial description of how the structure of
a cut-free proof is related to that of a proof with cut. Such
descriptions are important theoretical results which underlie
applications such as algorithmic cut-introduction as
in~\cite{Hetzl12Towards,HetzlXXAlgorithmic}.

This paper is an extended version of~\cite{Hetzl12Herbrand}, where we
have worked in a setting that was restricted to proofs of formulas of
the shape $\exists x_1 \cdots \exists x_n\, A$, for $A$
quantifier-free. In this paper we extend the results obtained
in~\cite{Hetzl12Herbrand} to proofs of arbitrary end-sequents. For
this, we first carry out the central technical work in a setting of
skolemized end-sequents, and then extend these results to the general
case via deskolemization. This proof strategy is analogous to the proof of
the second $\varepsilon$-Theorem from the first $\varepsilon$-Theorem
in~\cite{Hilbert39Grundlagen2}.  

More precisely, this paper is structured as follows: in
Section~\ref{sec.seqcalc_cutelim} we briefly review the sequent
calculus and cut-elimination for classical first-order logic. In
Section~\ref{sec.grammars} we describe regular and rigid tree grammars,
which we relate to proofs in Section~\ref{sec.proofs_grammars}. In
Section~\ref{sec.grammars_cutel} we prove the main invariance lemma in
the skolemized setting, and in Section~\ref{sec.skol_deskol} we
establish the necessary techniques and results for lifting the
invariance lemma to the general case. This lifting is carried out in
Section~\ref{sec.herbrand_content} followed by a discussion of several
corollaries such as Herbrand-confluence.


\section{Sequent Calculus and Cut-Elimination}\label{sec.seqcalc_cutelim}

For the sake of simplicity, we consider only a one-sided sequent
calculus and formulas in negation normal form, but the results can be
proved for a two-sided sequent calculus in the same way. Thus, our
\emph{formulas} (denoted by $A$, $B$, \ldots) are generated from the
literals and the constants $\top$ and $\bot$ via the binary
connectives $\land$ (\emph{and}) and $\lor$ (\emph{or}) and the
quantifiers $\exists$ and $\forall$ in the usual way. The
\emph{negation} $\dual{A}$ of a formula $A$ is defined via the usual
De~Morgan laws. A \emph{sequent} (denoted by $\Gamma$, $\Delta$,
\ldots) is a multiset of formulas.

\begin{defi}
A \emph{proof} is a tree of sequents, such that every
node forms together with its children an instance of one of the
inference rules shown below:
\begin{equation*}
\vlinf{}{\axr}{\;A,\dual{A}\;}{}
\qqquad
\vlinf{}{\top}{\;\top\;}{}
\qqquad
\vlinf{}{\weakr}{\Gamma, A}{\Gamma}
\qqquad
\vlinf{}{\conr}{\Gamma, A}{\Gamma, A, A}
\qqquad
\vliinf{}{\cutrr{}}{\Gamma,\Delta}{\Gamma, A\;}{\;\dual{A}, \Delta}
\end{equation*}
\begin{equation*}
\vlinf{}{\lor}{\Gamma, A\lor B}{\Gamma, A, B}
\qqquad
\vliinf{}{\land}{\Gamma,\Delta,A\land B}{\Gamma,A\;\;}{\;\;\Delta, B}
\qqquad
\vlinf{}{\forall}{\Gamma, \forall x\, A}{\Gamma, A\unsubst{x}{\alpha}}
\qqquad
\vlinf{}{\exists}{\Gamma, \exists x\, A}{\Gamma, A\unsubst{x}{t}}
\end{equation*}
where in the $\axr$-rule $A$ has to be a literal, in the
$\forall$-rule the $\alpha$ is called {\em eigenvariable} and does not
appear in $\Gamma,\forall x\, A$, and in the $\exists$-rule the term
$t$ does not contain a variable bound in $A$. We use the notation
$\unsubst{x}{\alpha}$ for the substitution that replaces $x$ by the
eigenvariable $\alpha$. Similarly, $\unsubst{x}{t}$ is the
substitution that replaces $x$ with $t$.
\end{defi}
The explicitly mentioned formula in a conclusion of an inference rule, like
$A\lor B$ for $\lor$ is called {\em main formula}. Analogously, the explicitly
mentioned formulas in the premises of an inference rule, like $A$ and $B$ for
$\lor$, are called {\em auxiliary formulas}. In the context of a concrete
derivation we speak about {\em main} and {\em auxiliary formula occurrences} of
inferences.

\begin{figure*}[!t]
  \footnotesize
  \def\myskip{0.2ex}
  $
    \begin{array}{r@{\qquad}r@{\qquad}c@{\qquad}l}
      \rlap{Axiom reduction:}\\[\myskip]
      \qqquad\qqquad&
      \vlderivation{
        \vliin{}{\cutrr{}}{\Gamma,A}{
          \vlhtr{\psi}{\Gamma,A}}{
          \vlin{}{\axr}{\dual{A},A}{
            \vlhy{}}}}
      &\scutred&
      \vlderivation{
        \vlhtr{\psi}{\Gamma,A}}
      \\\\
      \rlap{Quantifier reduction:}\\[\myskip]
      &
      \vlderivation{
        \vliin{}{\cutrr{}}{\Gamma,\Delta}{
          \vlin{}{\exists}{\Delta, \exists x\, \dual{A}}{
            \vlhtr{\psi_1}{\Delta, \dual{A}\unsubst{x}{t}}}}{
          \vlin{}{\forall}{\forall x\, A, \Gamma}{
            \vlhtr{\psi_2}{A\unsubst{x}{\alpha}, \Gamma}}}}
      &\scutred&
      \vlderivation{
        \vliin{}{\cutrr{}}{\Gamma,\Delta}{
          \vlbtrf{\psi_1}{\Delta, \dual{A}\unsubst{x}{t}}{2em}{1}}{
          \vlbtrf{\psi_2\unsubst{\alpha}{t}}{A\unsubst{x}{t},\Gamma}{
            4.5em}{.7}}}
      \\ \\
      \rlap{Propositional reduction:}\\[\myskip]
      &
      \vlderivation{
        \vliin{}{\cutrr{}}{\Gamma,\Delta,\Pi}{
          \vliin{}{\land}{\Gamma, \Delta, A\land B}{
            \vlhtr{\psi_1}{\Gamma,A}}{
            \vlhtr{\psi_2}{\Delta,B}}}{
          \vlin{}{\lor}{\dual{A}\lor \dual{B}, \Pi}{
            \vlhtr{\psi_3}{\dual{A},\dual{B},\Pi}}}}
      &\scutred&
      \vlderivation{
        \vliin{}{\cutrr{}}{\Gamma,\Delta,\Pi}{
          \vlhtr{\psi_2}{\Delta,B}}{
          \vliin{}{\cutrr{}}{\dual{B}, \Gamma, \Pi}{
            \vlhtr{\psi_1}{\Gamma,A}}{
            \vlhtr{\psi_3}{\dual{A},\dual{B},\Pi}}}}
      \\\\
      \rlap{Contraction reduction:}\\ 
      &
      \vlderivation{
        \vliin{}{\cutrr{}}{\Gamma,\Delta}{
          \vlin{}{\conr}{\Gamma, A}{
            \vlhtr{\psi_1}{\Gamma,A,A}}}{
          \vlhtr{\psi_2}{\dual{A},\Delta}}}
      &\scutred&
      \vlderivation{
        \vliq{}{\conr^\ast}{\Gamma,\Delta}{
          \vliin{}{\cutrr{}}{\Gamma,\Delta,\Delta}{
            \vliin{}{\cutrr{}}{\Gamma, \Delta, A}{
              \vlhtr{\psi_1}{\Gamma,A,A}}{
              \vlhtr{~\psi_2\rho'}{\dual{A},\Delta}}}{
            \vlhtr{~\psi_2\rho''}{\dual{A},\Delta}}}}
      \\\\
      \rlap{Weakening reduction:}\\[\myskip]
      &
      \vlderivation{
        \vliin{}{\cutrr{}}{\Gamma,\Delta}{
          \vlin{}{\weakr}{\Gamma, A}{
            \vlhtr{\psi_1}{\Gamma}}}{
          \vlhtr{\psi_2}{\dual{A},\Delta}}}
      &\scutred&
      \vlderivation{
        \vliq{}{\weakr^\ast}{\Gamma,\Delta}{
          \vlhtr{\psi_1}{\Gamma}}}
      \\\\
      \rlap{Unary inference permutation:}\\[\myskip]
      &
      \vlderivation{
        \vliin{}{\cutrr{}}{\Gamma,\Delta}{
          \vlin{}{\rr}{\Gamma,A}{
            \vlhtr{\psi_1}{\Gamma',A}}}{
          \vlhtr{\psi_2}{\dual{A},\Delta}}}
      &\scutred&
      \vlderivation{
        \vlin{}{\rr}{\Gamma,\Delta}{
          \vliin{}{\cutrr{}}{\Gamma',\Delta}{
            \vlhtr{\psi_1}{\Gamma',A}}{
            \vlhtr{\psi_2}{\dual{A},\Delta}}}}
      \\\\
      \rlap{Binary inference permutation:}\\[\myskip]
      &
      \vlderivation{
        \vliin{}{\cutrr{}}{\Gamma,\Delta}{
          \vliin{}{\rr}{\Gamma,A}{
            \vlhtr{\psi_1}{\Gamma'}}{
            \vlhtr{\psi_2}{\Gamma'',A}}}{
          \vlhtr{\psi_3}{\dual{A},\Delta}}}
      &\scutred&
      \vlderivation{
        \vliin{}{\rr}{\Gamma,\Delta}{
          \vlhtr{\psi_1}{\Gamma'}}{
          \vliin{}{\cutrr{}}{\Gamma'',\Delta}{
            \vlhtr{\psi_2}{\Gamma'',A}}{
            \vlhtr{\psi_3}{\dual{A},\Delta}}}}
    \end{array}
  $
  \caption{Cut-reduction steps\vadjust{\vskip-2.5ex}}
  \label{fig:cut-red}
\end{figure*}

\begin{defi}
  A proof is called {\em regular} if different $\forall$-inferences
  have different eigenvariables.
\end{defi}

We use the following convention: We use lowercase Greek letters
$\alpha,\beta,\gamma,\delta,\ldots$ for \emph{eigenvariables} in
proofs, and $\pi,\psi,\ldots$ for proofs.  For a proof $\pi$, we write
$|\pi|$ for the number of occurrences of inferences in~$\pi$.  Furthermore, we write
$\EV(\pi)$ for the set of eigenvariables of $\forall$-inferences of
$\pi$.

In a sequent calculus proof, each formula occurrence can be traced
downwards via its descendants to either a cut formula or the
end-sequent. We write $\EVc(\pi)$ for the set of those eigenvariables
in $\pi$ that are introduced by a $\forall$-inference whose main
formula occurrence can be traced downwards to a cut formula, i.e., is
not part of the end-sequent of~$\pi$. The elements of $\EVc(\pi)$ will
also be called \emph{cut-eigenvariables}.

\begin{defi}
  A \emph{weak sequent} is a sequent that does not contain any
  $\forall$-quantifier.
\end{defi}

\begin{fact}
  If the end-sequent of a proof $\pi$ is a weak sequent then
  $\EV(\pi)=\EVc(\pi)$.
\end{fact} 

\begin{rem}
  Our results do not depend on technical differences in the definition
  of the calculus (which in classical logic are inessential) such as
  the choice between multiplicative and additive rules and the
  differences in the cut-reduction induced by these choices. However,
  for the sake of precision, we will formally define the cut-reduction
  used in this paper.
\end{rem}

\begin{defi}
  Cut-reduction is defined on regular proofs and consists of the proof
  rewrite steps shown in Figure~\ref{fig:cut-red} (as well as all
  corresponding symmetric variants), where in the contraction
  reduction~step
  $$\rho'=\unsubst{\alpha}{\alpha'}_{\alpha\in \EV(\psi_2)}
  \qqquand
  \rho''=\unsubst{\alpha}{\alpha''}_{\alpha\in \EV(\psi_2)}$$ are
  substitutions replacing each eigenvariable $\alpha$
  in $\psi_2$ by fresh copies, i.e., $\alpha'$ and $\alpha''$ are
  fresh for the whole proof.
  We write $\cred$
  for the compatible (w.r.t.\ the inference rules), 
  reflexive and
  transitive closure of~$\scutred$.
\end{defi}

The above system for cut-reduction consists of purely local, minimal
steps and therefore allows the simulation of many other reduction relations.
We chose to work in this system in order to obtain invariance results of maximal strength.
Among the systems that can be simulated literally are for example all color annotations
of~\cite{Danos97New} in the multiplicative version of LK defined there. The real strength
of the results in this paper lies however in the general applicability of the
used proof techniques: the extraction of a grammar from a proof (that is described
in the next sections) is possible in all versions
of sequent calculus for classical logic and in principle also in other
systems like natural deduction. In particular, our results also apply to
inversion-based cut-elimination procedures such as for example that
in~\cite{Schwichtenberg77Proof}.


\section{Regular and Rigid Tree Grammars}\label{sec.grammars}

Formal language theory constitutes one of the main areas of
theoretical computer science.  Traditionally, a formal language is
defined to be a set of strings but this notion can be generalized in a
straightforward way to considering a language to be a set of
first-order terms. Such tree languages possess a rich theory and many
applications, see e.g.~\cite{Gecseg97Tree},~\cite{Comon07Tree}. In
this section we introduce notions and results from the theory of tree
languages that we will use for our proof-theoretic purposes.

A \emph{ranked alphabet} $\Sigma$ is a finite set of symbols which
have an associated arity (their \emph{rank}). For $f\in\Sigma$, we
sometimes use the notation $f/n$ for saying that $n$ is the arity
of~$f$. We write $\TR{\Sigma}$ to denote the set of all finite trees
(or terms) over $\Sigma$, and we write $\TRV{\Sigma}{X}$ to denote the
set of all trees over $\Sigma$ and a set $X$ of variables (seen as
symbols of arity 0).  We also use the notion of \emph{position} in a
tree, which is a list of natural numbers. We write $\varepsilon$ for
the empty list (the root position), and we write $p.q$ for the
concatenation of lists $p$ and~$q$. We write $p\le q$ if $p$ is a
prefix of $q$ and $p<q$ if $p$ is a proper prefix of $q$. Clearly,
$\le$ is a partial order and $<$ is its strict part. We write
$\Pos(t)$ to denote the set of all positions in a
tree~$t\in\TRV{\Sigma}{X}$. Furthermore, for a given tree or term~$t$
and position~$p$, we write $\atpos{t}{p}$ to denote the subterm of $t$
that occurs at position~$p$.

\begin{defi}
  A \emph{regular tree grammar} is a tuple $G=\tuple{N,\Sigma,\theta,P}$,
  where $N$ is a finite set of \emph{non-terminal symbols}, where
  $\Sigma$ is a ranked alphabet, such that $N\cap\Sigma=\emptyset$,
  where $\theta$ is the \emph{start symbol} with $\theta\in N$, and where $P$ is a finite
  set of production rules of the form $\beta\to t$ with $\beta \in N$ and 
  $t\in\TRV{\Sigma}{N}$.
\end{defi}

The derivation relation
$\derel{G}$ of a regular
tree grammar $G=\tuple{N,\Sigma,\theta,P}$ is defined as follows. We
have $s\derel{G} r$ if there is a production rule $\beta\to t$ in $P$
and a position $p\in\Pos(s)$, such that $\atpos{s}{p} = \beta$ and $r$
is obtained from $s$ by replacing $\beta$ at $p$ by $t$. The
\emph{language} of $G$ is then defined as
$\Lang{G}=\set{t\in\TR{\Sigma}\mid\theta\sderel{G}t}$, where
$\sderel{G}$ is the reflexive and transitive closure of $\derel{G}$. A
\emph{derivation $\DD$} of a term $t\in\Lang{G}$ is a sequence
$t_0\derel{G}t_1\derel{G}\ldots\derel{G}t_n$ with $t_0=\theta$ and
$t_n=t$. Note that a term $t$ might have different derivations in~$G$.

In~\cite{Jacquemard09Rigid} the class of rigid tree languages has been
introduced with applications in verification (e.g.\ of cryptographic
protocols as in~\cite{Jacquemard11Rigid}) as primary motivation. It
will turn out that this class is appropriate for describing
cut-elimination in classical first-order logic. In contrast
to~\cite{Jacquemard09Rigid} we do not use automata but grammars---their
equivalence is shown
in~\cite{Hetzl12Applying}.

\begin{defi}
  A \emph{rigid tree grammar} is a tuple
  $\tuple{N,N_R,\Sigma,\theta,P}$, where $\tuple{N,\Sigma,\theta,P}$, is a
  regular tree grammar and $N_R\subseteq N$ is the set of \emph{rigid
    non-terminals}.  We speak of a \emph{totally rigid tree grammar} if
  $N_R=N$. In this case we will just write
  $\tuple{N_R,\Sigma,\theta,P}$.
\end{defi}

A derivation $\theta=t_0\derel{G}t_1\derel{G}\ldots\derel{G}t_n=t$ of a
rigid tree grammar $G=\tuple{N,N_R,\Sigma,\theta,P}$ is a derivation in
the underlying regular tree grammar satisfying the additional
\emph{rigidity condition}: If there are $i,j<n$, a non-terminal
  $\beta\in N_R$, and positions $p$ and $q$ such that
  $\atpos{t_i}{p}=\beta$ and $\atpos{t_j}{q}=\beta$ then
  $\atpos{t}{p}=\atpos{t}{q}$.
The language $\Lang{G}$ of the rigid
tree grammar $G$ is the set of all terms $t\in\TR{\Sigma}$ which can
be derived under the rigidity condition.
For a given derivation
$\DD\colon\theta=t_0\derel{G}t_1\derel{G}\ldots\derel{G}t_n=t$ and a
non-terminal $\beta$ we say that $p\in\Pos(t)$ is a
\emph{$\beta$-position in~$\DD$} if there is an $i\le n$ with
$\atpos{t_i}{p}=\beta$, i.e., either a production rule
$\beta\rightarrow s$ has been applied at $p$ in $\DD$, or $\beta$
occurs at position~$p$ in~$t$.  In the context of a given grammar $G$,
we sometimes write $\DD\colon\alpha\sderel{G} t$ to specify that
$\DD$ is a derivation starting with $\alpha$ and ending with the
term~$t$.

\begin{exa}
Let $\Sigma=\{0/0,s/1\}$. A simple pumping argument shows that the
language $L=\{ f(t,t) \mid t\in\TR{\Sigma} \}$ is not regular. On the
other hand, $L$ is generated by the rigid tree grammar 
\begin{eqnarray*}
  G&=&\left\langle \set{\theta,\alpha,\beta}, \set{\alpha} , \set{ 0/0, s/1, f/2 }, \theta,P\right\rangle
  \quad\text{where}\\
  P&=&\{\; \theta \rightarrow f(\alpha,\alpha),\\
  &&\phantom{\{\;}\alpha\rightarrow 0 \mid s(\beta), \\
  &&\phantom{\{\;}\beta \rightarrow 0 \mid s(\beta)\; \}
\end{eqnarray*}
\end{exa}

\begin{lem}\label{lem.rigid_projection}
Let $G=\tuple{N,N_R,\Sigma,\theta,P}$ be a rigid tree grammar and let
$t\in\Lang{G}$. Then there is a derivation
$\theta\derel{G}\ldots\derel{G}t$ which uses at most one
$\beta$-production for each $\beta\in N_R$.
\end{lem}

\begin{proof}
Given any derivation of $t$, suppose
both $\beta\rightarrow s_1$ and $\beta\rightarrow s_2$ are
used at positions $p_1$ and $p_2$ respectively. Then by the rigidity
condition $t|_{p_1} = t|_{p_2}$ and we can replace the derivation at
$p_2$ by that at $p_1$ (or the other way round). This transformation
does not violate the rigidity condition because it only copies
existing parts of the derivation.
\end{proof}

\begin{lem}\label{lem.elim_rigid_NT}
Let $G=\tuple{N_R,\Sigma,\theta,P}$ be a totally rigid tree grammar and
$\theta\neq \beta \in N_R$, such that there is exactly one $t$ with
$\beta\rightarrow t$ in $P$. If
$G'=\tuple{N_R\setminus\set{\beta},\Sigma,\theta,
  (P\setminus\set{\beta\to t})\unsubst{\beta}{t}}$ then
$\Lang{G}=\Lang{G'}$.
\end{lem}
\begin{proof}
If a $G$-derivation of a term $s$ uses $\beta$, it must replace
$\beta$ by $t$ hence $s$ is derivable using the productions of $G'$ as
well. The rigidity condition is preserved as the equality constraints
of the $G'$-derivation are a subset of those of the $G$-derivation.
Conversely, given a $G'$-derivation of a term $s$ we obtain a
derivation of $s$ from the productions of $G$ by replacing
applications of $\delta\to r\unsubst{\beta}{t}$ by $\delta\to r$
followed by a copy of $\beta\to t$ for each occurrence of $\beta$ in
$r$. Let $\gamma_1,\ldots,\gamma_n$ be the non-terminals that appear
in $t$. By the rigidity condition for $i\in\{ 1,\ldots,n \}$ there is
a unique term at all $\gamma_i$-positions in the derivation. Hence
$\beta$ fulfills the rigidity condition as well, and we have obtained a
$G$-derivation of~$s$.
\end{proof}

%

\begin{nota}
For a given non-terminal $\beta$ and a term $t$, we will write
$\beta\in t$ or $t\ni\beta$ for denoting that $\beta$ occurs in~$t$.
\end{nota}

\begin{defi}
  Let $G$ be a tree grammar. A \emph{path} of $G$ is a list $\PP$ of
  productions $\alpha_1\to t_1$, \ldots, $\alpha_n\to t_n$ with
  $n\geq1$ and $\alpha_{i+1}\in t_i$ for all
  $i\in\set{1,\ldots,n-1}$. The \emph{length} of a path is
  $\length{\PP}=n$.  We will also write $\PP\colon\alpha_1\to
  t_1\ni\alpha_2\to\ldots\ni\alpha_n\to t_n$ to denote a path.
\end{defi}

For a given path $\PP\colon\alpha_1\to
t_1\ni\alpha_2\to\ldots\ni\alpha_n\to t_n$ we say that
$\alpha_1,\ldots,\alpha_n$ are \emph{on the path} $\PP$ and write
$\alpha_i\in\PP$ for that. We also write
$\PP\colon\alpha_1\pathto t_n$ and $\PP\colon\alpha_1\pathto\alpha_n$,
if we do not want to explicitly mention the intermediate steps.  For
a fixed grammar $G$, we write $\alpha\pathto\beta$ to denote that
there is a path $\PP$ in $G$ with $\PP\colon\alpha\pathto\beta$.  

For a set $P$ of production rules, we write $\alpha\reach{P}\beta$ (or
simply $\alpha\reach{}\beta$, when $P$ is clear from context) if there
is a production $\alpha\to t$ in $P$ with $\beta\in t$. We write
$\preach{}$ for the transitive closure of $\reach{}$, and $\sreach{}$
for its reflexive, transitive closure. Note that $\alpha\pathto\beta$
implies $\alpha\preach{}\beta$, but not the other way around, since
$\beta$ could be a non-terminal with no production $\beta\to s$
in~$P$.

\begin{defi}
  A tree grammar $\tuple{N,\Sigma,\theta,P}$ is called \emph{cyclic} if
  $\alpha\preach{P}\alpha$ for some $\alpha\in N$, and
  \emph{acyclic} otherwise.
\end{defi}

\begin{lem}\label{lem.totrig_acyclic_language}
  If $G$ is totally rigid and acyclic, then we have that up to renaming of the
  non-terminals
  $G=\tuple{\set{\alpha_1,\ldots,\alpha_n},\Sigma,\alpha_1,P}$ with
  $\Lang{G}=\set{\alpha_1\unsubst{\alpha_1}{t_1}\cdots\unsubst{\alpha_n}{t_n}
    \mid \alpha_i \rightarrow t_i \in P}$.
\end{lem}

\begin{proof}
  Acyclicity permits a renaming of non-terminals, such that
  $\alpha_i\preach{P}\alpha_j$ implies $i<j$. Then $\Lang{G}\supseteq
  \set{\alpha_1\unsubst{\alpha_1}{t_1}\cdots\unsubst{\alpha_n}{t_n}
    \mid \alpha_i \rightarrow t_i \in P}$ is obvious. For the
  left-to-right inclusion, let
  $\DD\colon\alpha_1=s_1\derel{G}\ldots\derel{G}s_n=s\in\TR{\Sigma}$
  be a derivation in $G$. By Lemma~\ref{lem.rigid_projection} we can
  assume that for each $j$ at most one production whose left-hand side
  is $\alpha_j$ is applied, say $\alpha_j\rightarrow t_j$. By
  acyclicity we can rearrange the derivation so that
  $\alpha_j\rightarrow t_j$ is only applied after $\alpha_i\rightarrow
  t_i$ for all $i<j$.  For those $\alpha_j$ which do not appear in the
  derivation we can insert any substitution without changing the final
  term so we obtain
  $s=\alpha_1\unsubst{\alpha_1}{t_1}\cdots\unsubst{\alpha_n}{t_n}$.
\end{proof}

This lemma entails that $|\Lang{G}| \leq \prod_{i=1}^n |\{ t \mid \alpha_i\rightarrow t \in P \}|$,
in particular we are dealing with a finite language. The central questions in this
context are (in contrast to the standard setting in formal language theory)
not concerned with {\em representability} but with the {\em size of a
representation}.


\section{Proofs and Grammars}\label{sec.proofs_grammars}

In this section we will relate sequent calculus proofs to rigid tree grammars.
A central tool for establishing this relation is Herbrand's
theorem~\cite{Herbrand30Recherches,Buss95Herbrand}. In its simplest form
it states that $\exists x\, A$, for $A$ quantifier-free, is valid iff
there are terms $t_1,\ldots, t_n$ such that $\Lor_{i=1}^n A\unsubst{x}{t_i}$
is a tautology. Such tautological disjunctions of instances are hence
called {\em Herbrand-disjunctions}. 
Such a disjunction, or equivalently: the set
of terms, can be considered a compact representation of a cut-free proof.
The relation to tree grammars is based on the observation that a (finite) set
of terms is just a (finite) tree language. While the Herbrand-disjunction of
a cut-free proof will be considered a tree language, a proof with cut will
give a rise to a grammar and its cut-elimination will be described by the
computation of the language of its grammar. 

There are different options for extending Herbrand's theorem to
non-prenex formulas, e.g.\ the Herbrand proofs
of~\cite{Buss95Herbrand} or the expansion trees
of~\cite{Miller87Compact}.  For our purposes it will be most useful to
follow the approach of~\cite{Baaz94Skolemization}.

\begin{defi}\label{def.Hseq}
  Let $\pi$ be a proof and let $O$ be a formula occurrence in
  $\pi$. Then we define the \emph{Herbrand-set} $\Hseq(O)$ of $O$
  inductively as follows:
  \begin{itemize}
  \item If $O$ is the occurrence of a formula $A$ in an axiom, then
    $\Hseq(O)=\set{A}$.
  \item If $O$ is in the conclusion sequent of an inference rule
    without being its main occurrence, then $O$ has exactly one
    ancestor $O'$ in one of the premises, and we let
    $\Hseq(O)=\Hseq(O')$.
  \item If $O$ is the main occurrence in the conclusion of a $\circ$-rule with
    $\circ\in\set{\land,\lor}$ and with auxiliary occurrences $O_1$
    and $O_2$, then $\Hseq(O)=\set{A\circ B\mid A\in\Hseq(O_1),
      B\in\Hseq(O_2)}$.
  \item If $O$ is the main occurrence in the conclusion of a
    $\forall$- or $\exists$-rule with auxiliary occurrence $O_1$ in the
    premise, then $\Hseq(O)=\Hseq(O_1)$.
  \item If $O$ is the main occurrence in the conclusion of a $\weakr$-rule, then
    $\Hseq(O)=\set{\bot}$.
  \item If $O$ is the main occurrence in the conclusion of a $\conr$-rule with auxiliary
    occurrences $O_1$ and $O_2$ in the premise, then
    $\Hseq(O)=\Hseq(O_1)\cup\Hseq(O_2)$.
  \end{itemize}
  Finally, we define $$\Hseq(\pi)=\bigcup_{P\in\Gamma}\Hseq(P)$$ 
  where
  $\Gamma$ is the end-sequent of $\pi$ and $P$ ranges over all formula
  occurrences in $\Gamma$.
\end{defi}

Besides to the Herbrand-set of a formula occurrence, we also need the
set of terms associated with an occurrence of an $\exists$-formula.

\begin{defi}
  Let $Q$ be an occurrence of a formula $\exists x\, A$ in a proof. We
  define the set $\terms{Q}$ of \emph{terms associated with $Q$} as
  follows: if $Q$ is introduced as the main formula of a weakening,
  then $\terms{Q}=\emptyset$. If $Q$ is introduced by an $\exists$-rule 
  \begin{equation*}
    \vlinf{}{\exists}{\Gamma, \exists x\, A}{\Gamma,A\unsubst{x}{t}}
  \end{equation*}
  then $\terms{Q}=\set{t}$. If $Q$ is the main formula in the
  conclusion of a contraction, and $Q_1$ and $Q_2$ are the two auxiliary
  occurrences of the same formula in the premise,
  then $\terms{Q}=\terms{Q_1}\cup\terms{Q_2}$. In all other cases, an
  inference with the occurrence $Q$ in the conclusion has a
  corresponding occurrence $Q'$ of the same formula in one of its
  premises, and we let $\terms{Q}=\terms{Q'}$.
\end{defi}

In the following, we will restrict our attention to a certain class of proofs,
that we call \emph{simple proofs} below.

\begin{defi}
A proof $\pi$ is called \emph{simple} if
\begin{itemize}
\item it is regular (i.e., different $\forall$-inferences
  have different eigenvariables),
\item every cut in $\pi$ is of
  one of the following forms:
  \begin{equation}
    \label{eq:simple}
    \vcenter{\lkcut{\Gamma,\Delta}{
        \Gamma, B
        &&
        \dual{B},\Delta
    }}
    \qqquor
    \vcenter{\lkcut{\Gamma,\Delta}{
        \Gamma,\exists x\, B
        &&
        \lku{\forall x\, \dual{B},\Delta}{
          \dual{B}\unsubst{x}{\alpha},\Delta
    }}}
  \end{equation}
  where $B$ is quantifier-free.
\end{itemize}
\end{defi}

Let us make some remarks on this definition. First, we require
regularity which is a necessary assumption in the context of
cut-elimination. But since every proof can be trivially transformed
into a regular one, this is no real restriction.  Second, the
requirement of the $\forall$-rule being applied directly above the cut
is natural as the rule is invertible. Moreover, any proof which does
not fulfill this requirement can be pruned to obtain one that does, by
simply permuting $\forall$-inferences down and identifying their
eigenvariables when needed. Thus, the only significant restriction is
that of disallowing quantifier alternations in the cut formulas. This
corresponds to allowing only $\Sigma_1$ (or $\Pi_1$) formulas in cuts.

We conjecture that our central result can be extended to
$\Sigma_n$-cuts. However, this will require the development of an
adequate class of grammars first (see also Section~\ref{sec.conclusion}).

\begin{obs}\label{obs:1-1}
  Simple proofs have the technically convenient property of exhibiting
  a 1-1 relationship between cut-eigenvariables and cuts.  For an
  eigenvariable $\alpha\in\EVc(\pi)$ we will therefore write
  $\cutrr{\alpha}$ for the corresponding cut and $\forallrr\alpha$ for
  the inference introducing $\alpha$ (when read from bottom to top).
\end{obs}

\begin{defi}
Let $\pi$ be a simple proof, let
$\alpha\in\EVc(\pi)$, and let $Q$ be the occurrence of the
existentially quantified cut-formula in the premise of
$\cutrr{\alpha}$. Then we write $\Bsub(\alpha)$ for the set of
substitutions $\set{\;\unsubst{\alpha}{t} \mid t\in\terms{Q}\;}$ and
we define $$\Bsub(\pi)=\Union_{\alpha\in\EVc(\pi)} \Bsub(\alpha)\quadfs$$
\end{defi}

Structures similar to the above $\Bsub(\pi)$ have been investigated
also in~\cite{Heijltjes10Classical} and~\cite{McKinley13Proof}
where they form the basis of proof net like formalisms using local
reductions for quantifiers in classical first-order logic.  Our aim in
this work is however quite different: we use these structures for a
global analysis of the sequent calculus.

\begin{defi}
Let $\pi$ be a simple proof. Then the
\emph{grammar of $\pi$} is the totally rigid
grammar $\Gram{\pi}=\tuple{N_R,\Sigma,\theta,P}$ with
\begin{align*}
  N_R&=\EVc(\pi)\cup\set{\theta}\\
  \Sigma&=\Sigma(\pi) \union \set{\land,\lor,\top,\bot}\\
  P&=\set{\theta\to A \mid A\in \Hseq(\pi)}\union 
\set{\alpha\to t \mid \unsubst{\alpha}{t}\in\Bsub(\pi)}
\end{align*}
where $\Sigma(\pi)$ is the signature of $\pi$,\footnote{We consider the
  eigenvariables in $\EV(\pi)\setminus\EVc(\pi)$ to be part of
  $\Sigma(\pi)$.} the rank of $\land$ and~$\lor$ is~2, the rank of
$\top$ and $\bot$ is~0, and $\theta$ does not occur in~$\pi$.
\end{defi}

\begin{lem}\label{lem.simple_acyclic}
  If $\pi$ is a simple proof, then $\Gram{\pi}$ is acyclic.
\end{lem}

\begin{proof}
  By induction on the number of cuts in $\pi$. The grammar of a
  cut-free proof is trivially acyclic. For the induction step, let
  $\rr$ be the lowest binary inference with subproofs $\pi_1$ and
  $\pi_2$ such that either (i) $\rr$ is a cut or (ii) $\rr$ is not a cut
  but both $\pi_1$ and $\pi_2$ contain at least one cut. Let $P$,
  $P_1$, and $P_2$ be the set of productions induced by the cuts in
  $\pi$, $\pi_1$, $\pi_2$, respectively. In case (ii),
  $\reach{P}=\reach{P_1}\cup\reach{P_2}$, which is acyclic by
  induction hypothesis (since
  $\EVc(\pi_1)\cap\EVc(\pi_2)=\emptyset$). In case (i), let $P_\rr$ be
  the productions induced by the cut $\rr$, then
  $\reach{P}=\reach{P_1}\cup\reach{P_2}\cup\reach{P_\rr}$. By
  induction hypothesis, $\reach{P_1}$ and $\reach{P_2}$ are acyclic
  and as the cut-formula in $\rr$ contains at most one quantifier,
  also $\reach{P_\rr}$ is acyclic. Therefore, a cycle in $\preach{P}$
  must be of the form
  $\alpha_1\sreach{P_1}\beta_1\nreach{P_\rr}\alpha_2\preach{P_2}\beta_2
  \nreach{P_\rr}\alpha_1$ where $\alpha_1,\beta_1\in\EVc(\pi_1)$ and
  $\alpha_2,\beta_2\in\EVc(\pi_2)$. However, $\rr$ contains only one
  quantifier and depending on its polarity all productions in
  $P_\rr$ lead from $\pi_1$ to $\pi_2$ or from $\pi_2$ to $\pi_1$
  but not both, so $\reach{P}$ is acyclic.
\end{proof}

\section{Grammars and Cut-Elimination}\label{sec.grammars_cutel}

In this section we will show that the language of the grammar of a
proof defined in the previous section is an invariant under cut
elimination. Before formulating this invariance result precisely we have to consider the
following three aspects of the situation:

First, note that all the reductions shown in Figure~\ref{fig:cut-red} preserve
simplicity, except the following:
\begin{equation*}
  \vlderivation{ 
    \vliiin{}{\cutrr{\beta}}{\qcdots}{
      \vliin{}{\cutrr{\alpha}}{\qcdots}{ 
        \vlhy{\qcdots}}{
        \vlin{}{\forallrr{\alpha}}{\qcdots}{
          \vlhy{\qcdots}}}}{
      \vlhy{\;\;}}{
      \vlin{}{\forallrr{\beta}}{\qcdots}{
        \vlhy{\qcdots}}}}
  \qqualto
  \vlderivation{
    \vliin{}{\cutrr{\alpha}}{\qcdots}{
      \vlhy{\qcdots}}{
      \vliiin{}{\cutrr{\beta}}{\qcdots}{
        \vlin{}{\forallrr{\alpha}}{\qcdots}{
          \vlhy{\qcdots}}}{
        \vlhy{\;\;}}{
        \vlin{}{\forallrr{\beta}}{\qcdots}{
          \vlhy{\qcdots}}}}}
\end{equation*}
where $\cutrr{\alpha}$ is permuted down under $\cutrr{\beta}$
(using the bottommost reduction in Fig.~\ref{fig:cut-red}) and the
cut formula of $\cutrr{\beta}$ has its ancestor on the right side of
$\cutrr{\alpha}$.
So in the following, when we speak about a
{\em reduction sequence of simple proofs} we require that the above
reduction is immediately followed by permuting $\forallrr{\alpha}$
down as well, in order to arrive at
\begin{equation*}
  \vlderivation{
    \vliin{}{\cutrr{\alpha}}{\qcdots}{
      \vlhy{\qcdots}}{
      \vlin{}{\forallrr{\alpha}}{\qcdots}{
        \vliiin{}{\cutrr{\beta}}{\qcdots}{
          \vlhy{\qcdots}}{
          \vlhy{\;\;}}{
          \vlin{}{\forallrr{\beta}}{\qcdots}{
            \vlhy{\qcdots}}}}}}
\end{equation*}
which is again simple.

Secondly, observe that there is no mechanism for deletion in the grammar, but
there is one in cut-elimination: the reduction of weakening which erases a
sub-proof (see Fig.~\ref{fig:cut-red}). It
is hence natural and will turn out to be useful to also consider the reduction relation without this step.

\begin{defi}
We define the \emph{non-erasing cut-reduction} $\credm$ as $\cred$ without
the reduction rule for weakening.
\end{defi}

Note that a $\credm$-normal form $\pi$ is an analytic proof
too as $\Hseq(\pi)$ is also a Herbrand-disjunction, i.e.\ a tautological
collection of instances. In contrast to a $\cred$-normal form
(which might contain implicit redundancy) a $\credm$-normal form might
also contain explicit redundancy in the form of cuts whose cut-formulas
are introduced by weakening on one or on both sides.  Non-erasing
reduction is also of interest in the context of the $\lambda$-calculus
where it is often considered in the form of the
$\lambda\mathrm{I}$-calculus and gives rise to the conservation
theorem (see Theorem 13.4.12 in~\cite{Barendregt84Lambda}).  Our
situation here is however quite different: neither $\cred$ nor
$\credm$ is confluent and neither of them is strongly normalizing.

Thirdly, in contrast to the case treated in~\cite{Hetzl12Herbrand} in
our more general setting it may happen that the reduction of a
weakening deletes sub-formulas of formula instances from the proof. In
order to treat this situation adequately, we need to define a
generalization of the $\subseteq$-relation between sets of
formulas. For this reason, we use the symbol~$\bot$ for representing
subformulas introduced by weakening, a technique also employed
in~\cite{Baaz12Complexity,Weller11Elimination} for the purpose of a
tighter complexity-analysis.

\begin{defi}
  The relation $\Hsub$ is defined inductively on quantifier-free
  formulas as follows:
  \begin{itemize}
  \item for all formulas $A$ we have $\bot\Hsub A$ and $A\Hsub A$, and
  \item whenever $A'\Hsub A$ and $B'\Hsub B$ then also 
    $A'\land B'\Hsub A\land B$ and $A'\lor B'\Hsub A\lor B$
  \end{itemize}
  Let $\cA$ and $\cB$ be sets of quantifier-free
  formulas. Then we define 
  \begin{equation*}
    \cA\Hsub\cB \qqquiff \text{for all $A\in\cA$ there
      is a $B\in\cB$ with $A\Hsub B$}\quad.    
  \end{equation*}
\end{defi}

\begin{fact}
  The relation $\Hsub$ is transitive on formula sets.
\end{fact}

We are now in a position to precisely state our main invariance lemma which
connects grammars with cut-elimination for weak sequents.

\begin{lem}\label{lem.inv_weak}
If $\pi\cred \pi'$ is a reduction sequence of simple proofs of a weak sequent, then
$\Lang{\Gram{\pi}}\Hsup \Lang{\Gram{\pi'}}$.
If $\pi\credm \pi'$ is a reduction sequence of simple proofs of a weak sequent,
then $\Lang{\Gram{\pi}}=\Lang{\Gram{\pi'}}$.
\end{lem}

The rest of this section is devoted to proving this result. The proof strategy
is to carry out an induction on the length of the reduction sequence
$\pi\cred \pi'$ (or $\pi\credm \pi'$ respectively) and to make a case distinction
on the type of reduction step. The most difficult step will turn out to be
the reduction of contraction which duplicates a sub-proof.

\begin{lem}\label{lem:H-invar}
  Let $\pi$ be a simple proof, and let $\pi'$ be
  obtained from $\pi$ by the single application of an axiom reduction,
  or a propositional reduction, or a unary or binary inference
  permutation (see Figure~\ref{fig:cut-red}). Then
  $\Lang{\Gram{\pi'}}=\Lang{\Gram{\pi}}$.
\end{lem}

\begin{proof}
  None of these reductions is changing the grammar of the proof, i.e.,
  $\Gram{\pi'}=\Gram{\pi}$ and therefore also
  $\Lang{\Gram{\pi'}}=\Lang{\Gram{\pi}}$.
\end{proof}

\begin{lem}\label{lem:H-invar-quant}
  Let $\pi$ be a simple proof, and let $\pi'$ be
  obtained from $\pi$ by the single application of a quantifier
  reduction (see Figure~\ref{fig:cut-red}). Then
  $\Lang{\Gram{\pi'}}=\Lang{\Gram{\pi}}$.
\end{lem}

\begin{proof}
  Let $\alpha$ be the eigenvariable of the $\forall$-inference and $t$ be
  the term of the $\exists$-rule directly above the cut that is
  reduced. Then $\Gram{\pi'}$ can be obtained from $\Gram{\pi}$ by
  removing the production rule $\alpha\to t$ and by applying the
  substitution $\unsubst{\alpha}{t}$ to the right-hand side of all remaining
  production rules. Thus, $\Lang{\Gram{\pi'}}=\Lang{\Gram{\pi}}$
  follows immediately from Lemma~\ref{lem.elim_rigid_NT}.
\end{proof}

\begin{lem}\label{lem:H-invar-weak}
    Let $\pi$ be a simple proof, and let $\pi'$ be
    obtained from $\pi$ by the single application of a weakening
    reduction (see Figure~\ref{fig:cut-red}). Then
    $\Lang{\Gram{\pi'}}\Hsub\Lang{\Gram{\pi}}$.
\end{lem}

\begin{proof}
  The grammar $\Gram{\pi'}$ is obtained from $\Gram{\pi}$ via two
  modifications. First, all productions coming from cuts or $\exists$-inferences in $\psi_2$
  are deleted, and second, the formulas in $\Delta$ which are ancestors of the end-sequent
  are replaced by $\bot$ in $\Hseq(\pi')$.
  Now let $A\in\Lang{\Gram{\pi'}}$. Then the derivation of $A$ in
  $\Gram{\pi'}$ is also a derivation in $\Gram{\pi}$, with the
  difference that some $\bot$-subformulas are replaced by other
  formulas, yielding a formula $B\in\Lang{\Gram{\pi}}$ with $B\Hsup
  A$. Hence $\Lang{\Gram{\pi'}}\Hsub\Lang{\Gram{\pi}}$.
\end{proof}

It remains to analyze the case of contraction. Surprisingly, also in
this case the language of the grammar of a proof remains
unchanged. However, the proof of this result is quite technical and
requires additional auxiliary results about the relationship between
proofs and grammars. Furthermore, this is the case which needs the
additional condition that the end-sequent of our proof is weak, i.e.,
does not contain $\forall$-quantifiers.

For simplifying the presentation, we assume in the
following (without loss of generality) that the $\forall$-side is on
the right of a cut and the $\exists$-side on the left. Then, a
production $\beta\rightarrow t$ in $\Gram{\pi}$ corresponds to three
inferences in $\pi$: a cut, an instance of the $\forall$-rule, and an
instance of the $\exists$-rule, that we denote by $\cutrr\beta$,
$\forallrr\beta$, and $\existsrr t$,
respectively, and that are, in general, 
arranged in $\pi$ as shown below.
\begin{equation}
\renewcommand{\vlvruler}{\vss\hbox{$\vdots$}\kern1pt}
  \vlderivation{
    \vliiin{}{\cutrr{\beta}}{\Gamma,\Delta}{
      \vlde{}{}{\Gamma,\vlex x A}{
        \vlin{}{\existsrr{t}}{\Gamma',\vlex x A}{
          \vlhy{\Gamma',A\unsubst{x}{t}}}}}{
      \vlhy{\qquad}}{
      \vlde{}{}{\vlfa x \dual A,\Delta}{
        \vlin{}{\forallrr\beta}{\vlfa x \dual A,\Delta'}{
          \vlhy{\dual A\unsubst{x}{\beta},\Delta'}}}}}
\end{equation}
The additional condition that $\forallrr\beta$ is directly above
$\cutrr\beta$, as indicated in~\eqref{eq:simple} is needed because
in the following we make extensive use of
Observation~\ref{obs:1-1}: there is a one-to-one correspondence
between the cuts and the eigenvariables in $\EVc(\pi)$, and thus, the
notation $\cutrr\beta$ makes sense.

\begin{defi}
  We say that the instances $\cutrr\beta$, $\forallrr\beta$, and
  $\existsrr t$ \emph{are on a path $\PP$ in $\Gram{\pi}$} if the
  production $\beta\to t$ is in $\PP$.
\end{defi}

\begin{defi}
  Let $\pi$ be a proof containing the configuration 
  \begin{equation*}
    \vlderivation{
      \vliiin{}{\rr_3}{\vdots}{
        \vlin{}{\rr_1}{\quad\ddots}{
          \vlhy{\vdots}}}{
        \vlhy{\quad}}{
        \vlin{}{\rr_2}{\iddots\quad}{
          \vlhy{\vdots}}}}
  \end{equation*}
  where $\rr_1$, $\rr_2$, and $\rr_3$ are arbitrary
  rule instances, and $\rr_3$ is a branching rule, and $\rr_1$ and
  $\rr_2$ might or might not be branching. Then we say that $\rr_1$ is
  \emph{on the left above} $\rr_3$, denoted by
  $\leftabove{\rr_1}{\rr_3}$, and $\rr_2$ is \emph{on the right above}
  $\rr_3$, denoted by $\rightabove{\rr_2}{\rr_3}$, and $\rr_1$ and
  $\rr_2$ are \emph{in parallel}, denoted by $\rparallel{\rr_1}{\rr_2}$.
\end{defi}

\begin{lem}\label{lem.path_cut}
  Let $\pi$ be a simple proof and $\PP\colon\alpha_1\rightarrow
  t_1\ni\alpha_2\ldots\rightarrow t_n$ be a path in $\Gram{\pi}$. Then
  there is a $k\in\set{1,\ldots,n}$ such that $\cutrr{\alpha_k}$ is
  lowermost among all inferences on $\PP$. Furthermore,
  $\forallrr{\alpha_1}$ is on the right above $\cutrr{\alpha_k}$ and
  $\existsrr{t_n}$ is on the left above $\cutrr{\alpha_k}$.
\end{lem}

\begin{proof}
  We proceed by induction on $n$. If $n=1$, then $n=k=1$. For the
  induction step consider a path $\alpha_1\to
  t_1\ni\ldots\ni\alpha_n\to t_n\ni\alpha_{n+1}\to t_{n+1}$.  By
  induction hypothesis, there is some $l\in\{1,\ldots,n\}$ such that
  we have this configuration
  \begin{equation*}
    \vlderivation{
        \vliiin{}{\cutrr{\alpha_l}}{\vdots}{
          \vlin{}{\existsrr{t_{n}}}{\quad\ddots}{
            \vlhy{\vdots}}}{
          \vlhy{\quad}}{
          \vlin{}{\forallrr{\alpha_1}}{\iddots\quad}{
            \vlhy{\vdots}}}}
  \end{equation*}
  As $\alpha_{n+1}\in t_n$ we know that
  $\existsrr{t_n}$ must be on the right above
  $\cutrr{\alpha_{n+1}}$. Hence, we are in one of the following two
  situations
  \begin{equation*}
    \vlderivation{
      \vliin{}{\cutrr{\alpha_{n+1}}}{\vdots}{
        \vlin{}{\existsrr{t_{n+1}}}{\quad\ddots}{
          \vlhy{\vdots}}}{
        \vliin{}{\cutrr{\alpha_l}}{\iddots\quad}{
          \vlin{}{\existsrr{t_{n}}}{\quad\ddots}{
            \vlhy{\vdots}}}{
          \vlin{}{\forallrr{\alpha_1}}{\iddots\quad}{
            \vlhy{\vdots}}}}}
    \nqquor
    \vlderivation{
      \vliin{}{\cutrr{\alpha_l}}{\vdots}{
        \vliin{}{\cutrr{\alpha_{n+1}}}{\quad\ddots}{
          \vlin{}{\existsrr{t_{n+1}}}{\quad\ddots}{
            \vlhy{\vdots}}}{
          \vlin{}{\existsrr{t_{n}}}{\iddots\quad}{
            \vlhy{\vdots}}}}{
        \vlin{}{\forallrr{\alpha_1}}{\iddots\quad}{
          \vlhy{\vdots}}}}
  \end{equation*}
  In the first case we let $k=n+1$ and in the second we let $k=l$. In both cases 
  $\cutrr{\alpha_k}$ has the desired properties.
\end{proof}

\begin{lem}\label{lem.path_cutpos}
  Let $\pi$ be a simple proof,
  let $\Gram{\pi}=\tuple{N_R,\Sigma,\varphi,P}$, and
  let $\beta,\alpha\in\EVc(\pi)$.  If $\beta\pathto{}\alpha$ then either
  $\leftabove{\cutrr{\alpha}}{\cutrr{\beta}}$ or
  $\rightabove{\cutrr{\beta}}{\cutrr{\alpha}}$ or 
  $\rparallel{\cutrr{\alpha}}{\cutrr{\beta}}$.
\end{lem}

\begin{proof}
  Since $\beta\pathto{}\alpha$, we have a path
  $\beta\to\ldots\ni\alpha\to t$ for some $t$.  By
  Lemma~\ref{lem.path_cut} there is a $\gamma$, such that
  $\leftabove{\existsrr t}{\cutrr{\gamma}}$ and 
  $\rightabove{\forallrr{\beta}}{\cutrr{\gamma}}$, and such that
 $\cutrr{\alpha}$ and $\cutrr{\beta}$ are not below
$\cutrr\gamma$. Furthermore, $\cutrr{\alpha}$ must be below
$\existsrr t$, and $\cutrr{\beta}$ below $\forallrr{\beta}$.
If $\gamma=\beta$, then $\leftabove{\cutrr{\alpha}}{\cutrr{\beta}}$. 
If $\gamma=\alpha$, then $\rightabove{\cutrr{\beta}}{\cutrr{\alpha}}$. 
And if $\gamma\neq\beta$ and $\gamma\neq\alpha$, then 
$\rparallel{\cutrr{\alpha}}{\cutrr{\beta}}$.
\end{proof}

\begin{lem}\label{lem.alphapos}
  Let $\Gram{\pi}=\tuple{N_R,\Sigma,\varphi,P}$ be the grammar of a
  simple proof $\pi$, such that there are two paths
  \begin{eqnarray*}
    &&\beta\to t\ni\gamma_0\to s_0\ni\gamma_1\to s_1\ni\ldots
    \to s_{n-1}\ni\gamma_n=\alpha\to s_n\\
    &&\beta\to t\ni\delta_0\to r_0\ni\delta_1\to r_1\ni\ldots
    \to r_{m-1}\ni\delta_m=\alpha\to r_m
  \end{eqnarray*}
  such that $\gamma_0$ and $\delta_0$ occur at two different positions
  in~$t$. Then we have
  one of the following two cases:
  \begin{enumerate}
  \item we have $\gamma_i=\delta_j$ for some $0\le i<n$ and $0\le j<m$, or \vadjust{\vskip.7ex}
  \item for all $0\le i<n$ and $0\le j<m$ we have\/
    $\rightabove{\cutrr{\gamma_i}}{\cutrr{\alpha}}$ and\/
    $\rightabove{\cutrr{\delta_j}}{\cutrr{\alpha}}$.
  \end{enumerate}
\end{lem}

\begin{proof}
  Note that because of acyclicity of $\Gram{\pi}$, we have that
  $\beta\neq\gamma_i$ for all $i\le n$ and $\beta\neq\delta_j$ for all
  $j\le m$, in particular $\beta\neq\alpha$. Assume, for the moment,
  that $m,n>0$; the case of one of them being zero will be treated at
  the very end of the proof. Then $\gamma_0\neq\alpha$ and
  $\delta_0\neq\alpha$.  If $\gamma_0=\delta_0$, we have case~1. So,
  assume also $\gamma_0\neq\delta_0$.  As $\beta\rightarrow t$ is a
  production in $\Gram{\pi}$, the proof $\pi$ contains a formula which
  contains both $\gamma_0$ and $\delta_0$ hence $\forallrr{\gamma_0}$
  and $\forallrr{\delta_0}$ are not parallel. Since we have
  $\rightabove{\forallrr{\gamma_0}}{\cutrr{\gamma_0}}$ and
  $\rightabove{\forallrr{\delta_0}}{\cutrr{\delta_0}}$, we also have
  that $\cutrr{\gamma_0}$ and $\cutrr{\delta_0}$ are not
  parallel. Without loss of generality, assume that $\cutrr{\delta_0}$
  is below $\cutrr{\gamma_0}$. Then
  $\rightabove{\cutrr{\gamma_0}}{\cutrr{\delta_0}}$ (since
  $\leftabove{\cutrr{\gamma_0}}{\cutrr{\delta_0}}$ would entail
  $\rparallel{\forallrr{\gamma_0}}{\forallrr{\delta_0}}$). Since we
  have $\delta_0\pathto{}\alpha$, we can apply
  Lemma~\ref{lem.path_cutpos}, giving us three possibilities:
  \begin{itemize}
  \item If $\leftabove{\cutrr{\alpha}}{\cutrr{\delta_0}}$ then we have
    the situation
    \begin{equation*}
      \vlderivation{
      \vliin{}{\cutrr{\delta_0}}{\vdots}{
        \vliin{}{\cutrr{\alpha}}{\quad\ddots}{
          \vlin{}{\existsrr{s_n}}{\quad\ddots}{
            \vlhy{\vdots}}}{
          \vlin{}{\forallrr{\alpha}}{\iddots\quad}{
            \vlhy{\vdots}}}}{
        \vliin{}{\cutrr{\gamma_0}}{\iddots\quad}{
          \vlin{}{\existsrr{s_0}}{\quad\ddots}{
            \vlhy{\vdots}}}{
          \vlin{}{\forallrr{\gamma_0}}{\iddots\quad}{
            \vlhy{\vdots}}}}}
    \end{equation*}
    By Lemma~\ref{lem.path_cut} applied to the path
    $\gamma_0\pathto{}s_n$ we have that $\cutrr{\delta_0}$ must
    coincide with $\cutrr{\gamma_i}$ for some $0\le i< n$ (since $\pi$
    is a tree), so $\delta_0=\gamma_i$ (by Observation~\ref{obs:1-1}),
    and we are in case~1.
  \item If $\rightabove{\cutrr{\delta_0}}{\cutrr{\alpha}}$ then we are
    in \emph{both} of the following two situations:
    \begin{equation*}
      \vlderivation{
      \vliin{\qquad}{\cutrr{\alpha}}{\vdots}{
        \vlin{}{\existsrr{s_n}}{\quad\ddots}{
          \vlhy{\vdots}}}{
        \vliin{}{\cutrr{\delta_0}}{\iddots\quad}{
          \vlhy{\quad\ddots}}{
          \vliin{}{\cutrr{\gamma_0}}{\iddots\quad}{
            \vlhy{\quad\ddots}}{
            \vlin{}{\forallrr{\gamma_0}}{\iddots\quad}{
              \vlhy{\vdots}}}}}}
      \nqquand
      \vlderivation{
      \vliin{}{\cutrr{\alpha}}{\vdots}{
        \vlin{}{\existsrr{r_m}}{\quad\ddots}{
          \vlhy{\vdots}}}{
        \vliin{}{\cutrr{\delta_0}}{\iddots\quad}{
          \vlhy{\quad\ddots}}{
          \vlin{}{\forallrr{\delta_0}}{\iddots\quad}{
              \vlhy{\vdots}}}}}
     \end{equation*}
    Thus, by Lemma~\ref{lem.path_cut} applied to the paths
    $\gamma_0\pathto{}s_n$ and $\delta_0\pathto{}r_m$ we know that
    $\cutrr{\alpha}=\cutrr{\gamma_k}=\cutrr{\delta_l}$ for some $0\le
    k\le n$ and $0\le l\le m$ hence
    $\gamma_k=\alpha=\delta_l$. Furthermore $k=n$ and $l=m$ by
    acyclicity of $\Gram{\pi}$ and assumption $\gamma_n = \alpha =
    \delta_m$.
    Now consider any $\gamma_i$ with $0\le
    i<n$. Since $\gamma_i\pathto{}\alpha$, we can apply
    Lemma~\ref{lem.path_cutpos} and get either
    $\leftabove{\cutrr{\alpha}}{\cutrr{\gamma_i}}$ or
    $\rightabove{\cutrr{\gamma_i}}{\cutrr{\alpha}}$ or
    $\rparallel{\cutrr{\alpha}}{\cutrr{\gamma_i}}$. Since by
    Lemma~\ref{lem.path_cut} $\cutrr{\gamma_i}$ must be above
    $\cutrr{\alpha}$, we conclude
    $\rightabove{\cutrr{\gamma_i}}{\cutrr{\alpha}}$. With the same
    reasoning we can conclude that
    $\rightabove{\cutrr{\delta_j}}{\cutrr{\alpha}}$ for all $0\le
    j<m$. We are therefore in case~2.
  \item If $\rparallel{\cutrr{\alpha}}{\cutrr{\delta_0}}$ then we are
    in \emph{both} of the following two situations:
    \begin{equation*}
      \vlderivation{
        \vliin{}{\rr}{\vdots}{
          \vliin{}{\cutrr{\alpha}}{\quad\ddots}{
            \vlin{}{\existsrr{r_m}}{\quad\ddots}{
              \vlhy{\vdots}}}{
            \vlhy{\iddots\quad}}}{          
          \vliin{}{\cutrr{\delta_0}}{\iddots\quad}{
            \vlhy{\quad\ddots\quad}}{
            \vlin{}{\forallrr{\delta_0}}{\iddots\quad}{
              \vlhy{\vdots}}}}}
    \end{equation*}
    and\semiproofadjust
    \begin{equation*}
      \vlderivation{
        \vliin{}{\rr}{\vdots}{
          \vliin{}{\cutrr{\alpha}}{\quad\ddots}{
            \vlin{}{\existsrr{s_n}}{\quad\ddots}{
              \vlhy{\vdots}}}{
            \vlhy{\iddots\quad}}}{
          \vliin{}{\cutrr{\delta_0}}{\iddots\quad}{
            \vlhy{\quad\ddots}}{
            \vliin{}{\cutrr{\gamma_0}}{\iddots\quad}{
              \vlhy{\quad\ddots\quad}}{
              \vlin{}{\forallrr{\gamma_0}}{\iddots\quad}{
                \vlhy{\vdots}}}}}}
    \end{equation*}
    By Lemma~\ref{lem.path_cut} applied to the paths $\gamma_0\to
    \ldots \to s_n$ and $\delta_0\to \ldots \to r_m$, the rule $\rr$
    coincides with $\cutrr{\gamma_i}$ and $\cutrr{\delta_j}$ for some
    $0<i<n$ and $0<j<m$, therefore $\gamma_i=\delta_j$ (by
    Observation~\ref{obs:1-1}), and we are in case~1.
  \end{itemize}
  It remains to treat the case $n=0$ or $m=0$. If $m=n=0$ then we are
  trivially in case~2 (there is no $0\le i<n$ or $0\le j<m$). If
  $n=0$ and $m>0$, we can apply Lemma~\ref{lem.path_cut} to the path
  $\delta_0\to \ldots \to r_m$ and obtain an $l\in\set{0,\ldots,m}$
  such that we are in the situation
  \begin{equation*}
      \vlderivation{
        \vliin{}{\cutrr{\delta_l}}{\vdots}{
          \vliin{}{\cutrr{\alpha}}{\quad\ddots}{
            \vlin{}{\existsrr{r_m}}{\quad\ddots}{
              \vlhy{\vdots}}}{
            \vlin{}{\forallrr{\alpha}}{\iddots\quad}{
              \vlhy{\vdots}}}}{          
          \vliin{}{\cutrr{\delta_0}}{\iddots\quad}{
            \vlhy{\quad\ddots\quad}}{
            \vlin{}{\forallrr{\delta_0}}{\iddots\quad}{
              \vlhy{\vdots}}}}}
  \end{equation*}
  But by the same argument as at the beginning of the proof, we also have
  that $\forallrr{\alpha}$ and $\forallrr{\delta_0}$ cannot be in
  parallel ($\alpha$ and $\delta_0$ both appear in $t$), and therefore
  either $\rightabove{\cutrr{\alpha}}{\cutrr{\delta_0}}$ or
  $\rightabove{\cutrr{\delta_0}}{\cutrr{\alpha}}$. Since
  $\delta_0\pathto{}\alpha$, the only possibility is
  $\rightabove{\cutrr{\delta_0}}{\cutrr{\alpha}}$, by
  Lemma~\ref{lem.path_cutpos}. Thus
  $\cutrr{\alpha}=\cutrr{\delta_l}$, and therefore $l=m$ and we are in
  case~2. The case $m=0$ and $n>0$ is similar.
\end{proof}

We have now finally collected together all necessary tools for describing
the reduction step for contraction.

\begin{lem}\label{lem.subset_contraction_reduction}
Let $\pi$ be a simple proof of a weak sequent such that $\pi$ contains
a subproof $\psi$, shown on the left below,
\begin{equation*}
  \psi
  \quad = \quad
  \vlderivation{
    \vliin{}{\cutr}{\Gamma,\Delta}{
      \vlin{}{\conr}{\Gamma,A}{
        \vlbtrf{\psi_1}{\Gamma,A,A}{2em}{.8}}}{
      \vlbtrf{\psi_2}{\dual A,\Delta}{2em}{.8}}}
  \qqualto
  \vlderivation{
    \vliq{}{\conr^\ast}{\Gamma,\Delta}{
      \vliin{}{\cutr}{\Gamma,\Delta,\Delta}{
        \vliin{}{\cutr}{\Gamma,\Delta,A}{
          \vlbtrf{\psi_1}{\Gamma,A,A}{2em}{.8}}{
          \vlbtrf{\psi_2\rho'}{\dual A,\Delta}{2em}{.8}}}{
        \vlbtrf{\psi_2\rho''}{\dual A,\Delta}{2em}{.8}}}}
  \quad = \quad
  \psi'
\end{equation*}
and let $\pi'$ be the proof obtained from $\pi$ from replacing $\psi$
by $\psi'$ shown on the right above, where
$\rho'=\unsubst{\alpha}{\alpha'}_{\alpha\in \EV(\psi_2)}$ and
$\rho''=\unsubst{\alpha}{\alpha''}_{\alpha\in \EV(\psi_2)}$ are
substitutions that replace all eigenvariables in $\psi_2$ by fresh
copies.  Then $\Lang{\Gram{\pi'}}=\Lang{\Gram{\pi}}$.
\end{lem}

\begin{proof}
Let us first show
$\Lang{\Gram{\pi}}\subseteq\Lang{\Gram{\pi'}}$. Write $P$ for the
productions of $\Gram{\pi}$ and $P'$ for those of $\Gram{\pi'}$. Let
$F\in\Lang{\Gram{\pi}}$ and $\DD$ be its derivation.  If the
duplicated cut is quantifier-free, then $P' = P\rho' \union P\rho''$,
since the substitutions $\rho$ and $\rho'$ do not affect the
eigenvariables outside $\psi_2$. Hence $\DD\rho'$ (as well as
$\DD\rho''$) is a derivation of $F$ in $\Gram{\pi'}$. If the
duplicated cut contains a quantifier, let $\alpha$ be its
eigenvariable, let $t_1,\ldots,t_k$ be its terms coming from the left
copy of $A$ and $t_{k+1},\ldots,t_n$ those from the right copy of $A$
and let $Q=\{ \alpha\to t_1,\ldots,\alpha\to t_n \}\subseteq P$. We
then have
\begin{equation*}
  P'=(P\setminus Q)\rho'\union\{ \alpha'\to t_1,\ldots,\alpha'\to t_k \}
  \union 
  (P\setminus Q)\rho'' \union \{ \alpha''\to t_{k+1},\ldots,
  \alpha''\to t_n \}\quad.
\end{equation*}
If $\DD$ does not contain $\alpha$, then $\DD\rho'$ (as well as
$\DD\rho''$) is a derivation of $F$ in $\Gram{\pi'}$. If $\DD$ does
contain $\alpha$, then by Lemma~\ref{lem.rigid_projection} we can
assume that it uses only one $\alpha$-production, say $\alpha\to
t_i$. If $1\leq i \leq k$, then $\DD\rho'$ is a derivation of $F$
in~$\Gram{\pi'}$ and if $k< i\leq n$, then $\DD\rho''$ is a derivation
of $F$ in~$\Gram{\pi'}$.

Let us now show $\Lang{\Gram{\pi'}}\subseteq\Lang{\Gram{\pi}}$. Let
$F$ be a formula in $\Lang{\Gram{\pi'}}$, and let $\DD'$ be a
derivation of $F$ in $\Gram{\pi'}$. We construct
$\DD=\DD'(\rho')^{-1}(\rho'')^{-1}$ by ``undoing'' the renaming of the
variables in $\psi_2$. Then $\DD$ is a derivation for $F$, using the
production rules of $\Gram{\pi}$, but possibly violating the rigidity
condition.

First, recall that $\EVc(\pi)=\EV(\pi)$ and observe that only
non-terminals $\alpha\in\EV(\psi_2)$ can violate the rigidity
condition in $\DD$: if $\beta\notin\EV(\psi_2)$ violates the rigidity
condition then there are $\beta$-positions $p_1,p_2$ in $\DD$ with
$F|_{p_1}\neq F|_{p_2}$ and as $\beta\rho'\rho''=\beta$ the positions
$p_1,p_2$ are also $\beta$-positions in $\DD'$ and they violate the
rigidity condition in $\DD'$ which is a contradiction to $\DD'$ being
a $\Gram{\pi'}$-derivation.

Now define for each $\alpha\in\EV(\psi_2)$ the value
$\nonrig{\DD}{\alpha}$ to be the number of pairs
$(p_1,p_2)\in\Pos(F)\times\Pos(F)$ where $p_1$ and $p_2$ are
$\alpha$-positions in $\DD$ with $p_1\neq p_2$ and $F|_{p_1}\neq
F|_{p_2}$, and define
$\nonrigd{\DD}=\sum_{\alpha\in\EV(\psi_2)}\nonrig{\DD}{\alpha}$. We
proceed by induction on
$\nonrigd{\DD}$ to show that $\DD$
can be transformed into a derivation which does no longer violate
rigidity. If $\nonrigd{\DD}=0$ then
$\DD$ obeys the rigidity condition, and we are done. Otherwise there
is at least one $\alpha\in\EV(\psi_2)$ with
$\nonrig{\DD}{\alpha}>0$. We now pick one such $\alpha$ which is
minimal with respect to $\sreach{}$ (which exists since
$\Gram{\pi}$ is acyclic). Let $p_1$ and $p_2$ be $\alpha$-positions in
$\DD$ with $p_1\neq p_2$ and $F|_{p_1}\neq F|_{p_2}$, let $p$ be
the maximal common prefix of $p_1$ and $p_2$ and let $q$ be the
maximal prefix of $p$ where a production rule has been applied in $\DD$. Due to the tree
structure of $F$, the position $q$ is uniquely defined, and $q$ is a
$\beta$-position for some non-terminal $\beta$, and some production
rule $\beta\to t$ has been applied at position $q$ in $\DD$, and we
have two paths:
\begin{equation*}
  \begin{array}{c}
    \beta\to t\ni\gamma_0\to s_0\ni\gamma_1\to s_1\ni\ldots
    \to s_{n-1}\ni\gamma_n=\alpha\to s_n
    \\[1ex]
    \beta\to t\ni\delta_0\to r_0\ni\delta_1\to r_1\ni\ldots
    \to r_{m-1}\ni\delta_m=\alpha\to r_m
  \end{array}
\end{equation*}
where $\gamma_0$ and $\delta_0$ occur at two different positions
in~$t$. Thus, we can apply Lemma~\ref{lem.alphapos}, giving us the
following two cases:
\begin{itemize}
\item We have $\gamma_i=\delta_j$ for some $0\le i<n$ and $0\le
  j<m$. Say $\eta=\gamma_i=\delta_j$, and let $p_{\gamma}$ and
  $p_{\delta}$ be the positions of $\gamma_i$ and $\delta_j$
  (respectively) in $\DD$. Since $\eta\preach{}\alpha$ we know that
  $\eta$ does not violate the rigidity condition (we chose $\alpha$ to
  be minimal), and therefore $F|_{p_{\gamma}}=F|_{p_\delta}=F'$. Let
  $\DD_\gamma\colon\gamma_i\sderel{\Gram{\pi}}F'$ and
  $\DD_\delta\colon\delta_j\sderel{\Gram{\pi}}F'$ be the two
  subderivations of $\DD$ starting in positions $p_{\gamma}$ and
  $p_{\delta}$, respectively. Without loss of generality, we can
  assume that $\nonrigd{\DD_\gamma}\le\nonrigd{\DD_\delta}$. Then let
  $\DDt$ be the derivation obtained from $\DD$ by replacing
  $\DD_\delta$ by $\DD_\gamma$. Then $\DDt$ is still a derivation for
  $F$, but $\nonrigd{\DDt}<\nonrigd{\DD}$.
\item For all $0\le i<n$ and $0\le j<m$ we have
  $\rightabove{\cutrr{\gamma_i}}{\cutrr{\alpha}}$ and
  $\rightabove{\cutrr{\delta_j}}{\cutrr{\alpha}}$. So all inferences
  of the path $\gamma_0\to \ldots \to s_{n-1}$ as well as all
  inferences of $\delta_0\to\ldots\to r_{m-1}$ are in
  $\psi_2$. Therefore all variables of of these paths are in
  $\EV(\psi_2)$. As $\alpha$ violates the rigidity in $\DD$ one of
  $p_1,p_2$ must be a $\alpha'$-position and the other a
  $\alpha''$-position in $\DD'$ because $\DD'$ does satisfy the
  rigidity condition. Without loss of generality we can assume
  that~$p_1$ is the $\alpha'$-position and $p_2$ the
  $\alpha''$-position. As the paths are contained completely in
  $\psi_2$ we have $\gamma_0\in\EV(\psi_2)\rho'$ and
  $\delta_0\in\EV(\psi_2)\rho''$ which is a contradiction as no term
  can contain both a variable from $\EV(\psi_2)\rho'$ and one from
  $\EV(\psi_2)\rho''$.  \qedhere
\end{itemize}
\end{proof}

\begin{proof}[Proof of Lemma~\ref{lem.inv_weak}]
By induction on the length of the reduction $\pi\cred\pi'$
or $\pi\credm\pi'$ respectively using one of Lemmas~\ref{lem:H-invar}, 
\ref{lem:H-invar-quant}, \ref{lem:H-invar-weak} or \ref{lem.subset_contraction_reduction}
depending on the current reduction step.
\end{proof}


\section{Skolemization and Deskolemization}\label{sec.skol_deskol}

In this section we will describe some results that allow one to extend the above invariance lemma
to proofs of arbitrary end-sequents
(including $\forall$-quantifiers). Carrying out the above argument directly 
for arbitrary end-sequents would require dealing with variable-names on
the level of the grammar in order to describe the changes of eigenvariables
of the $\forall$-quantifiers in the end-sequent. This can be avoided completely
by skolemizing proofs to reduce the general case to that of weak sequents and then
translating back the results by deskolemization. Skolemization and deskolemization
are simple operations on the level of Herbrand-disjunctions or
expansion trees~\cite{Miller87Compact} and their use in this context suffices
for our purposes. In contrast, they have surprising complexity-effects
on the level of proofs, see e.g.~\cite{Baaz12Complexity}. The reason why this transfer
is possible is that the form of the end-sequent,
and in particular the question whether it contains universal quantifiers,
does not have an effect on the dynamics of cut-elimination. This 
observation has been well known for a long time and is apparent already in Gentzen's consistency
proof for Peano Arithmetic~\cite{Gentzen38Neue} which is carried out on a (hypothetical) proof of
the empty sequent as well as in the proof of the second $\varepsilon$-Theorem
from the first $\varepsilon$-Theorem by
deskolemization~\cite{Hilbert39Grundlagen2}.

Let us now first define the notion of Herbrand-disjunction precisely. We assume
w.l.o.g.\ that in a formula every variable is bound by at most one quantifier.

\begin{defi}
For a given formula $F$, we write $\hat{F}$ for the formula obtained from
$F$ by removing all quantifiers. Now let $x_1,\ldots,x_n$ be the
existentially bound variables in $F$, and let $y_1,\ldots,y_m$ be the
universally bound variables in $F$. Then any formula of the shape 
$$\doublehat{F}\sop
\sel{x_1}{t_1},\ldots,\sel{x_n}{t_n},\sel{y_1}{\alpha_1},\ldots,\sel{y_m}{\alpha_m}\scl$$
where $\doublehat{F}$ is an arbitrary formula with $\doublehat{F}
\Hsub \hat{F}$, where $t_1,\ldots,t_n$ are arbitrary terms, and where
$\alpha_1,\ldots,\alpha_m$ are fresh variables, is called an
\emph{instance of $F$}.
If $\Gamma$ is a sequent we say that a set $\cI$ of formulas is a \emph{set
of instances of $\Gamma$} if for every $I \in \cI$ there is a $F\in\Gamma$, s.t.\ $I$ is instance
of $F$.
\end{defi}

Often we will work in the context of a proof $\pi$ of a sequent $\Gamma$ and
consider the instances of the formulas in $\Gamma$ that are induced by $\pi$.
Then the above fresh variables $\alpha_1,\ldots,\alpha_m$ will be eigenvariables
of the proof and their occurrences in terms will be restricted by an acylicity-condition, see below.

Let $\Gamma = F_1,\ldots,F_n$ be a sequent, let $\cI$ be a set of
instances of $\Gamma$, let $m_i$ be the number of quantifiers in
$F_i$, and let $l_i$ be the number of instances of $F_i$ in $\cI$. If
we impose an arbitrary linear ordering on the instances of $F_i$ in
$\cI$, then a tuple $\langle i, j, k \rangle$ for $1\leq i \leq n$ and
$1 \leq j \leq m_i$ and $1 \leq k \leq l_i$ uniquely identifies the
term which is substituted for the quantifier $Q x_j$ in the $k$-th
instance of the formula $F_i$. We will write $t_{i,j,k}$ for this term
(which could just be an eigenvariable if $Q x_j$ happens to be an
$\forall$-quantifier).  The $k$-th instance of $F_i$ can hence be
written as $F_{i,k}\sop
\sel{x_1}{t_{i,1,k}},\ldots,\sel{x_{m_i}}{t_{i,m_i,k}} \scl$, where
$x_1,\ldots,x_{m_i}$ are the bound variables in $F_i$, and $F_{i,k}$
is some formula with $F_{i,k}\Hsub\hat{F_i}$.  Such a tuple $\langle
i, j, k\rangle$ is called \emph{existential position} if $x_j$ is
bound existentially in $F_i$, and \emph{universal position} if $x_j$
is bound universally in $F_i$.

A position $\langle i_1,j_1,k_1 \rangle$ is said to \emph{dominate}
another position $\langle i_2,j_2,k_2 \rangle$, if $i_1 = i_2$, and
$k_1 = k_2$, and the quantifier $Q x_{j_2}$ is in the scope of the
quantifier $Q x_{j_1}$ in $F_i$. A set $\cI$ of instances induces a
relation $\deprelz$ on its existential positions as: $\langle
i_1,j_1,k_1 \rangle \deprelz \langle i_2,j_2,k_2 \rangle$ if there is
a universal position $\langle i_3,j_3,k_3 \rangle$, such that the term
$t_{i_2,j_2,k_2}$ contains a variable $\alpha$ with $\alpha =
t_{i_3,j_3,k_3}$ and $\langle i_1,j_1,k_1 \rangle$ dominates $\langle
i_3,j_3,k_3 \rangle$. Furthermore we define the \emph{dependency
  relation $\deprel$} on the existential positions of $\cI$ as
transitive closure of~$\deprelz$.

\begin{rem}
  A proof $\pi$ with the property that $\Hseq(\pi)=\cI$ is sometimes
  called a \emph{sequentialization of $\cI$}.  If $\cI$ has positions
  $\langle i_1,j_1,k_1 \rangle$ and $\langle i_2,j_2,k_2 \rangle$ with
  $\langle i_1,j_1,k_1 \rangle \deprelz \langle i_2,j_2,k_2 \rangle$,
  then in each sequentialization of $\cI$ the inference corresponding
  to $\langle i_1, j_1, k_1\rangle$ is below that of $\langle i_2,
  j_2, k_2\rangle$. In the literature on proof nets, relations like
  $\deprelz$ are known as \emph{jumps}.
\end{rem}

\begin{defi}
  A set $\cI$ of instances of $\Gamma$ is called
  \emph{Herbrand-disjunction of $\Gamma$} if
  \begin{itemize}
  \item the dependency relation $\deprel$ of $\cI$ is acyclic, and
  \item $\displaystyle\Lor_{I\in\cI} I$ is a tautology.
  \end{itemize}
\end{defi}

This notion of Herbrand-disjunction is essentially a flat (as opposed to tree-like)
formulation of expansion tree proofs~\cite{Miller87Compact}. A similar flat formulation
can, for instance, be found in~\cite{Baaz94Skolemization}.

\begin{theorem}
$\Gamma$ is valid iff it has a Herbrand-disjunction.
\end{theorem}
\begin{proof}[Proof Sketch]
Via translating back and forth with cut-free sequent calculus or alternatively via
expansion tree proofs.
\end{proof}
\begin{exa}\label{ex.drinker}
Let $\Gamma = \exists x\, (\dual{P}(x) \lor \forall y\, P(y))$, let
$\cI =\set{\dual{P}(c)\lor P(\alpha), \dual{P}(\alpha)\lor P(\beta)$} and
fix the numbering of quantifiers and instances to be from the left to the
right. Then there are the two existential positions $\langle 1, 1, 1\rangle$
with $t_{1,1,1} = c$ and $\langle 1,1,2 \rangle$ with $t_{1,1,2} = \alpha$
and two universal positions $\langle 1,2,1 \rangle$ with $t_{1,2,1} = \alpha$
and $\langle 1,2,2\rangle$ with $t_{1,2,2} = \beta$. As $\langle 1,1,k \rangle$
dominates $\langle 1,2,k \rangle$, we have $\langle 1,1,1 \rangle \deprelz \langle 1,1,2\rangle$,
but not the other way round because $t_{1,1,1} = c$ is variable-free. Therefore
$\deprel$ is acyclic. Furthermore $\cI$ is a tautology and hence a Herbrand-disjunction.
\end{exa}
Note that for a weak sequent $\Gamma$, the induced dependency ordering $\deprel$ is empty and
hence trivially acyclic. The Herbrand-disjunctions of weak sequents are therefore exactly
the tautologies of instances.
\begin{defi}
Let $F[\forall y\, G]$ be a formula containing a universal quantifier
and let $\exists x_1$, \ldots, $\exists x_n$ be the existential quantifiers in whose scope
$\forall y$ is. Then define the \emph{Skolemization} of this universal quantifier~as
\[
\sk_1(F[\forall y\, G]) = F[G\unsubst{y}{g(x_1,\ldots,x_n)}
\]
where $g$ is a fresh $n$-ary function symbol, called a {\em Skolem
  function symbol}.  The term $g(x_1,\ldots,x_n)$ is called
\emph{Skolem-term}.  For a formula $F$ define its \emph{Skolemization
  $\sk(F)$} to be the iteration of $\sk_1$ until no universal
quantifier is left, such that no Skolem function symbol is used for
two different universal quantifiers in $F$. For a sequent $\Gamma =
F_1,\ldots,F_n$ define its \emph{Skolemization $\sk(\Gamma) =
  \sk(F_1),\ldots,\sk(F_n)$}, where no Skolem function symbol is used
for two different universal quantifiers in $\Gamma$.
\end{defi}

\begin{rem}
Sometimes the above operation on formulas is also called Herbrandization. We
prefer to use the name Skolemization due to the simple duality between the
satisfiability-preserving replacement of existential quantifiers and the
validity-preserving replacement of universal quantifiers by new function
symbols. There is no danger of confusion as, in the proof-theoretic context
of this work, we are clearly dealing with validity only. This use of terminology
is due to~\cite{Hilbert39Grundlagen2}, see in particular Section~3.5.a.
\end{rem}

The above side condition on the choice of Skolem function symbols results in
a 1-1 mapping between universal quantifiers
in the sequent we skolemize and the Skolem function symbols. It could be
made formally more precise by equipping the $\sk$-operation with such a bijection
as second argument. However, for the sake of notational simplicity we refrain
from doing so here.

The Skolemization of formulas and sequents can be extended to a Skolemization of proofs. When
skolemizing a proof, all universal quantifiers in the end-sequent are removed and their variables are replaced
by Skolem-terms. In contrast, the cut-formulas remain unchanged, more precisely:
\begin{defi}
  Let $\pi$ be a proof of a sequent $\Gamma$, and let $y_1,\ldots,y_n$
  be the variables that are bound by a $\forall$-quantifier in
  $\Gamma$. Furthermore, for each $y_i$ let
  $\alpha_{i,1},\ldots,\alpha_{i,h_i}$ be the eigenvariables
  introduced in $\pi$ by an $\forall$-rule whose main formula is of the
  shape $\forall y_i\, A$.  Then the \emph{Skolemization of the proof $\pi$},
  denoted by $\sk(\pi)$, is the proof with end-sequent $\sk(\Gamma)$
  that is obtained from $\pi$ by 
  \begin{enumerate}
  \item removing all $\forall$-quantifiers binding one of
    $y_1,\ldots,y_n$ everywhere, and
  \item replacing each occurrence of $y_i$ (for
    $i\in\set{1,\ldots,n}$) and $\alpha_{i,j}$ (for
    $i\in\set{1,\ldots,n}$ and $j\in\set{1,\ldots,h_i}$) by the
    corresponding Skolem-term. This term is in each case uniquely
    determined if we proceed from the end-sequent of $\pi$ upwards to
    the axioms and demand that each rule application remains valid,
    or, in the case of the $\forall$-rule, becomes void (i.e., premise
    and conclusion coincide), and
  \item removing the void rule instances.
  \end{enumerate}
  Note that $\sk(\pi)$ still can contain $\forall$-quantifiers, namely
  those coming from a cut.
\end{defi}

The Skolemization of a proof $\pi$ also affects the quantifier-free
formulas in $\pi$ through the replacement of eigenvariables by Skolem
terms. In the context of proof Skolemization we hence extend the notation
$\sk(\cdot)$ to formulas $F$ from which some (or all) $\forall$-quantifiers
have been removed; then $\sk(F)$ denotes the formula obtained from
skolemizing the remaining $\forall$-quantifiers \emph{and} carrying out
the replacement of eigenvariables by Skolem-terms.
Skolemization of proofs has the following useful commutation properties.

\begin{lem}\label{lem.sk_cred_commute}
If $\pi \cred \pi'$ then $\sk(\pi)\cred\sk(\pi')$. If $\pi\credm \pi'$ then
$\sk(\pi)\credm\sk(\pi')$.
\end{lem}
\begin{proof}
By induction on the number of reductions in $\pi\cred\pi'$ or $\pi\credm\pi'$,
respectively, making a case distinction on the reduction step. The most
interesting case is that of the permutation of a $\forall$-inference over a cut
\[
      \vlderivation{
        \vliin{}{\cutrr{}}{\Gamma, \forall x\, B, \Delta}{
          \vlin{}{\forall}{\Gamma, \forall x\, B,A}{
            \vlhtr{\psi_1}{\Gamma, B\unsubst{x}{\alpha},A}}}{
          \vlhtr{\psi_2}{\dual{A},\Delta}}}
      \qquad\scutred \qquad
      \vlderivation{
        \vlin{}{\forall}{\Gamma, \forall x\, B, \Delta}{
          \vliin{}{\cutrr{}}{\Gamma, B\unsubst{x}{\alpha}, \Delta}{
            \vlhtr{\psi_1}{\Gamma, B\unsubst{x}{\alpha}, A}}{
            \vlhtr{\psi_2}{\dual{A},\Delta}}}}
\]
where the main formula of the $\forall$-inference is an ancestor of the end-sequent.
This reduction step is translated to an identity-step as Skolemization maps
both of the above proofs to
\[
      \vlderivation{
          \vliin{}{\cutrr{}}{\sk(\Gamma), \sk(\forall x\, B), \sk(\Delta)}{
            \vlhtr{\psi^\mathrm{s}_1}{\sk(\Gamma), \sk(\forall x\, B), \sk(A)}}{
            \vlhtr{\psi^\mathrm{s}_2}{\sk(\dual{A}), \sk(\Delta)}}}
\]
Each of the other reduction steps translates directly into exactly one reduction
step in the skolemized sequence.
\end{proof}

\begin{lem}\label{lem.sk_LangGram_commute}
$\Lang{\Gram{\sk(\pi)}} = \sk(\Lang{\Gram{\pi}})$.
\end{lem}
\begin{proof}\sloppy
First note that $\EVc(\pi) = \EVc(\sk(\pi))$ hence $\Gram{\pi}$ and
$\Gram{\sk(\pi)}$ have the same non-terminals. Furthermore, to each
$\alpha\in\EV(\pi)\setminus\EVc(\pi)$ corresponds a unique Skolem-term
in $\sk(\pi)$, hence to each $F\in\Hseq(\pi)$ and $\sigma\in\Bsub(\pi)$ corresponds
a unique $F'\in\Hseq(\sk(\pi))$ and $\sigma'\in\Bsub(\pi)$ and therefore to
each production $\alpha\to t$ in $\Gram{\pi}$ corresponds a unique production
$\alpha\rightarrow t'$ in $\Gram{\sk(\pi)}$ that is obtained from replacing
eigenvariables by their respective Skolem-terms.
If $I \in \sk(\Lang{\Gram{\pi}})$ then by
Lemma~\ref{lem.totrig_acyclic_language} we have $I =
\sk(F\unsubst{\alpha_1}{s_1}\cdots\unsubst{\alpha_n}{s_n})$. Now for
$\theta\to F, \alpha_1\to s_1,\ldots,\alpha_n\to s_n$ being the
productions in $\Gram{\pi}$, letting $\theta \to F', \alpha_1\to
s'_1,\ldots,\alpha_n\to s'_n$ be the corresponding productions in
$\Gram{\sk(\pi)}$ we obtain
$F'\unsubst{\alpha_1}{s'_1}\cdots\unsubst{\alpha_n}{s'_n} =
\sk(F\unsubst{\alpha_1}{s_1}\cdots\unsubst{\alpha_n}{s_n})$. Thus,
$\sk(\Lang{\Gram{\pi}})\subseteq\Lang{\Gram{\sk(\pi)}}$.
For the other direction, note that every Skolem-term has at least one
corresponding $\alpha\in\EV(\pi)\setminus \EVc(\pi)$, and as before,
this relation translates to productions. So, if $J\in
\Lang{\Gram{\sk(\pi)}}$ then by
Lemma~\ref{lem.totrig_acyclic_language} we have $J =
G\unsubst{\alpha_1}{t_1}\cdots\unsubst{\alpha_n}{t_n}$ for $\theta
\to G, \alpha_1 \to t_1,\ldots, \alpha_n\to t_n$ being the productions
in $\Gram{\sk(\pi)}$.  By choosing one corresponding set of
productions $\theta\to G',\alpha_1\to t'_1,\ldots,\alpha_n\to t'_n$
where Skolem-terms are replaced by the eigenvariables from which they
originate we obtain
$\sk(G'\unsubst{\alpha_1}{t'_1}\cdots\unsubst{\alpha_n}{t'_n}) =
G\unsubst{\alpha_1}{t_1}\cdots\unsubst{\alpha_n}{t_n}$.
\end{proof}

As we have seen in the above proof, Skolemization can identify instances that
differ only in their variable names. The reason for this ability lies in the use
of variable names which can be chosen in a redundant way. These superfluous instances
can also be removed by an appropriate variable renaming as shown in the following
example.

\begin{exa}
Let $\Gamma = \exists x\forall y\, ( \dual{P}(x,y) \lor \dual{Q}(x,y) ),
\exists x\, P(c,x) \land \exists x\, Q(c,x)$. Then the set of instances
obtained from a sequent calculus proof that ends with an $\land$-inference is
\[
\cI = \set{\dual{P}(c,\alpha)\lor\dual{Q}(c,\alpha), \dual{P}(c,\beta)\lor\dual{Q}(c,\beta),
P(c,\alpha)\land Q(c,\beta)}\quadfs
\]
Skolemizing would produce the following set of instances
\[
\sk(\cI) = \set{\dual{P}(c,f(c))\lor\dual{Q}(c,f(c)), P(c,f(c))\land Q(c,f(c))}
\]
by implicitly identifying the two formulas that become equal. A similar effect
(but without using Skolemization) can be achieved by directly
identifying $\alpha$ and $\beta$ as in
\[
\cI\unsubst{\beta}{\alpha} = \set{\dual{P}(c,\alpha)\lor\dual{Q}(c,\alpha), P(c,\alpha)\land Q(c,\alpha)}\quadfs
\]
\end{exa}

We now generalize the observations made in the above example.  For
every Herbrand-disjunction $\cI$ there is a substitution $\rho$, such
that $\cI\rho$ is a Herbrand-disjunction having the following
property: If two universal positions $\langle i,j, k_1\rangle$ and
$\langle i,j,k_2\rangle$ have different variables then there is a
$j'$, such that the quantifier $\exists x_{j'}$ dominates $\forall
x_j$ in $F_i$ and $t_{i,j',k_1} \neq t_{i,j',k_2}$. This follows for
example from the formulation of expansion trees
in~\cite{Chaudhuri12Systematic,ChaudhuriXXIsomorphism} which use sets
of terms for the $\exists$-quantifier and a single variable for the
$\forall$-quantifier.  A Herbrand-disjunction with this property is
$\alpha$-equivalent to one with {\em canonical variable names} in the
following sense.
\begin{defi}
Let $\cI$ be a set of instances. The {\em canonical name} of the eigenvariable
of the universal position $\langle i,j,k \rangle$
is $\alpha_{i,j,t_1,\ldots,t_m}$ where $t_1,\ldots,t_m$ are the terms of the existential
positions that dominate $\langle i,j,k\rangle$. The {\em canonical variable
renaming} $\rho_\mathrm{c}$ of $\cI$ is the substitution which replaces all variable
names by their canonical names.
\end{defi}

\begin{rem}
  Note that this relationship is significantly more complex than
  $\alpha$-equivalence, as differently named variables are identified
  according to certain criteria external to variable names. In
  particular, for some fixed $\cI$, there are $\cI_n$ of unbounded
  size such that $\cI_n\rho_\mathrm{c} = \cI$.  This can be seen, for
  example, by continuing Example~\ref{ex.drinker}: take $\cI_n = \{
  \dual{P}(c)\lor P(\alpha_i), \dual{P}(\alpha_i) \lor P(\beta_i) \mid
  1\leq i\leq n \}$.
\end{rem}

We now turn to \emph{deskolemization}, the inverse operation of Skolemization. In our
setting, we only consider deskolemization of sequents and their instances,
but not of proofs. Furthermore we always assume that the original sequent
with $\forall$-quantifiers is known. Hence the deskolemization of a sequent
trivially replaces it by the original sequent. More
interesting is the deskolemization of instances which will consist of
replacing Skolem-terms by (canonically named) variables.

\begin{defi}
Let $\Gamma = F_1,\ldots, F_n$ be a sequent with Skolem function symbol $f_{i,j}$
for the universal quantifier $\forall x_j$ in $F_i$. Let $\cI$ be a set of
instances of $\Gamma$ and define its \emph{deskolemization} $\sk^{-1}(\cI)$ by
repeating the replacement
\[
f_{i,j}(t_1,\ldots,t_m) \mapsto \alpha_{i,j,t_1,\ldots,t_m}
\]
on maximal Skolem-terms (w.r.t.\ the subterm ordering).
\end{defi}
In the deskolemization of a Herbrand-disjunction, the acyclicity of
the dependency relation is obtained from the acyclicity of the subterm
ordering on the Skolem-terms. Conversely, during Skolemization, the
Skolem-terms are well-defined due to the acyclicity of the dependency
relation (see
e.g.~\cite{Miller87Compact,Weller11Elimination,Baaz12Complexity} for
more details).  We hence obtain the following properties:

\begin{lem}\label{lem.desk_canonic}
Let $\Gamma$ be a sequent and $\Gamma'$ be a weak sequent with $\Gamma'=\sk(\Gamma)$.
\begin{enumerate}
\item If $\cI$ is a Herbrand-disjunction of $\Gamma$, then $\sk(\cI)$ is
a Herbrand-disjunction of $\sk(\Gamma)$.
\item If $\cI'$ is a Herbrand-disjunction of $\Gamma'$, then $\sk^{-1}(\cI')$
is a Herbrand-disjunction of $\sk^{-1}(\Gamma')$.
\item If $\cI$ is a Herbrand-disjunction of $\Gamma$, then
$\sk^{-1}(\sk(\cI)) = \cI\rho_\mathrm{c}$.
\end{enumerate}
\end{lem}


\section{Herbrand-Content}\label{sec.herbrand_content}

\begin{defi}
For a simple proof $\pi$, we define its
\emph{Herbrand-content} as $\Hcont{\pi} = \Lang{\Gram{\pi}}\rho_\mathrm{c}$.
\end{defi}

Note that for a cut-free proof $\pi$ we have $\Hcont{\pi} =
\Hseq(\pi)\rho_\mathrm{c}$, i.e.\ the Herbrand-content is nothing other than
the Herbrand-disjunction of the proof after variable
normalization. Also note that for a proof $\pi$ of a weak sequent we
have $\Hcont{\pi} = \Lang{\Gram{\pi}}$, and hence, for a cut-free proof
of a weak sequent we have $\Hcont{\pi} = \Hseq(\pi)$. We can now lift the
main invariance lemma, Lemma~\ref{lem.inv_weak}, to proofs of
arbitrary end-sequents and formulate this result in terms of the
Herbrand-content.

\begin{theorem}\label{thm.main}
If $\pi\cred \pi'$ is a reduction sequence of simple proofs, then
$\Hcont{\pi}\Hsup \Hcont{\pi'}$.
If $\pi\credm \pi'$ is a reduction sequence of simple proofs,
then $\Hcont{\pi}=\Hcont{\pi'}$.
\end{theorem}

\begin{proof}
If $\pi\cred \pi'$ then $\sk(\pi)\cred \sk(\pi')$ by Lemma~\ref{lem.sk_cred_commute}.
So, by Lemma~\ref{lem.inv_weak}, we have $\Lang{\Gram{\sk(\pi)}} \Hsup \Lang{\Gram{\sk(\pi')}}$.
By Lemma~\ref{lem.sk_LangGram_commute}, we get
$\sk(\Lang{\Gram{\pi}}) \Hsup \sk(\Lang{\Gram{\pi'}})$. Using
Lemma~\ref{lem.desk_canonic} and the observation that $\sk^{-1}$ commutes with $\Hsub$
we see that
\[
\Hcont{\pi} = \Lang{\Gram{\pi}}\rho_\mathrm{c} = \sk^{-1}(\sk(\Lang{\Gram{\pi}})) \Hsup
\sk^{-1}(\sk(\Lang{\Gram{\pi'}})) = \Lang{\Gram{\pi'}}\rho_\mathrm{c} = \Hcont{\pi'}
\]
The proof for $\pi\credm\pi'$ is step-by-step the same, replacing $\Hsup$
by $=$.
\end{proof}

\begin{cOr}\label{cor.upperbound}
If $\pi\cred\pi'$ is a reduction sequence of simple proofs and
$\pi'$ is cut-free, then $$\Hseq(\pi')\rho_\mathrm{c}\Hsub \Hcont{\pi}\quadfs$$
\end{cOr}
\begin{proof}
This is a direct consequence of Theorem~\ref{thm.main}.
\end{proof}

This corollary shows that $\Hcont{\pi}$ is an upper bound on the Herbrand-disjunctions obtainable by
cut-elimination from $\pi$. Let us now compare this result with
another upper bound that has previously been obtained
in~\cite{Hetzl10Form}. To that aim let $\Gramreg{\pi}$ denote the
regular tree grammar underlying $\Gram{\pi}$ which can be obtained by
setting all non-terminals to non-rigid. In this notation, a central
result of~\cite{Hetzl10Form}, adapted to this paper's setting is

\begin{theorem}\label{thm.formwitness}
Let $\pi$ be a proof of a formula of the shape $\exists
x_1\ldots\exists x_n\, A$ with $A$ quantifier-free, and let $\pi\cred
\pi'$ with $\pi'$ cut-free. Then $\Hseq(\pi')\subseteq
\Lang{\Gramreg{\pi}}$.
\end{theorem}

While the Theorem~\ref{thm.formwitness} applies also to non-simple proofs,
Corollary~\ref{cor.upperbound} is stronger in several respects:

First, the size of the Herbrand-content is by an exponential smaller
than the size of the bound given by Theorem~\ref{thm.formwitness}.
Indeed, it is a straightforward consequence of
Lemma~\ref{lem.totrig_acyclic_language} that the language of a totally rigid
acyclic tree grammar with $n$ production rules is bound by $n^n$ but on
the other hand:

\begin{prop}
There is an acyclic regular tree grammar $G$ with $2n$ productions and $|\Lang{G}|=n^{n^{n}}$.
\end{prop}

\begin{proof}
Let $f$ be an $n$-ary function symbol, then the productions
$\alpha_0\rightarrow f(\alpha_1,\ldots,\alpha_1)$, \ldots, $\alpha_{n-1}\rightarrow f(\alpha_n,\ldots,\alpha_n)$
create a tree with $n^n$ leaves. Let $c_1,\ldots,c_n$ be terminal symbols, then by adding
the productions
$\alpha_n\rightarrow c_1,\ldots,\alpha_n\rightarrow c_n$
we obtain the desired grammar $G$.
\end{proof}

Secondly, the class of totally rigid acyclic tree grammars can be
shown to be in exact correspondence with the class of simple proofs in
the following sense.  Not only can we use a totally rigid acyclic tree
grammar to simulate the process of cut-elimination, we can also---in
the other direction---use cut-elimination to simulate the process of
calculating the language of a grammar. It is shown
in~\cite{Hetzl12Applying} how to transform an arbitrary acyclic
totally rigid tree grammar $G$ into a simple proof that has a $\cred$
normal form whose Herbrand-disjunction is essentially the language of
$G$.

The third and---for the purposes of this paper---most important
difference is that the bound of Corollary~\ref{cor.upperbound} is
{\em tight} in the sense that it can actually be reached by a cut-elimination
strategy, namely $\credm$. In fact, an even stronger statement is true: not only is
there a normal form of $\credm$ that reaches the bound but all of them do.
This property leads naturally to the following confluence result for
classical logic.

\begin{defi}[Herbrand-confluence]
A relation $\genrel$ on a set of proofs is called \emph{Herbrand-confluent}
if $\pi\genrel\pi_1$ and $\pi\genrel\pi_2$ with $\pi_1$ and $\pi_2$ being
normal forms for $\genrel$ implies that
$\Hseq(\pi_1)\rho_\mathrm{c}=\Hseq(\pi_2)\rho_\mathrm{c}$.
\end{defi}

\begin{cOr}\label{cor.Hconfluence}
The relation $\credm$ is Herbrand-confluent on the set of simple proofs.
\end{cOr}
\begin{proof}
This is a direct consequence of Theorem~\ref{thm.main}.
\end{proof}

How does this result fit together with $\credm$ being neither
confluent nor strongly normalizing? In fact, note that it is possible to construct a simple proof which permits
an infinite $\credm$ reduction sequence from which one can obtain
normal forms of arbitrary size by bailing out from time to time. This can be done by building on the propositional double-contraction
example found e.g.\ in~\cite{Danos97New,Gallier93Constructive,Urban00Classical} and
in a similar form in~\cite{Zucker74Correspondence}. While these infinitely many normal forms do have pairwise different Herbrand-disjunctions
when regarded as {\em multisets}, Corollary~\ref{cor.Hconfluence} shows that as {\em sets} they are all the same.
This set-character of Herbrand-disjunctions is assured by using canonical variable
names (or equivalently: Skolemization) and thus identifying repeated instances.
This observation shows that the lack of strong normalization is taken care of
by using sets instead of multisets as data structure. But what about the lack of
confluence? Results like~\cite{Baaz11Nonconfluence} and~\cite{Hetzl12Computational} show
that the number of $\cred$ normal forms with different Herbrand-disjunctions
can be enormous. On the other hand we have just seen that $\credm$ induces only
{\em a single} Herbrand-disjunction: $\Hcont{\pi}$. The relation between $\Hcont{\pi}$
and the many Herbrand-disjunctions induced by $\cred$ is explained by Corollary~\ref{cor.upperbound}:
$\Hcont{\pi}$ contains them all.

\section{Conclusion}\label{sec.conclusion}

We have shown that non-erasing cut-elimination for the class of simple
proofs is Herbrand-confluent. While there are different and possibly
infinitely many normal forms, they all induce the same
Herbrand-disjunction. This result motivates the definition of this
unique Herbrand-disjunction as Herbrand-{\em content} of the proof
with cut.

As future work, the authors plan to extend this result to arbitrary
first-order proofs. The treatment of blocks of quantifiers is straightforward:
the rigidity condition must be changed to apply to vectors of non-terminals.
Treating quantifier alternations is more difficult: the current
results suggest to use a \emph{stack} of totally rigid tree grammars,
each layer of which corresponds to one layer of quantifiers (and is
hence acyclic). Concerning further generalizations, note that the 
method of describing a cut-free proof by a tree language is
applicable to any proof system with quantifiers that has a
Herbrand-like theorem, e.g., even full higher-order logic as
in~\cite{Miller87Compact}. The difficulty consists in finding an
appropriate type of grammars.

Given the wealth of different methods for the extraction of
constructive content from classical proofs, what we learn from our work
about the class of simple proofs is this: the
first-order structure possesses (in contrast to the propositional
structure) a unique and canonical unfolding. The various extraction
methods hence do not differ in the choice of how to unfold the
first-order structure but only in choosing {\em which part} of it to
unfold. We therefore see that the effect of the underspecification of
algorithmic detail in classical logic is redundancy.

\section*{Acknowledgments}
The authors would like to thank Paul-Andr\'{e} Melli\`{e}s for helpful comments
on this work. The first author was supported by a Marie Curie Intra European
Fellowship within the 7th European Community Framework Programme, by the
projects I603, P22028 and P25160 of the Austrian Science Fund (FWF) and
the WWTF Vienna Research Group 12-04.


\bibliographystyle{alpha}
\bibliography{references}

\end{document}